\documentclass[journal]{IEEEtran}

\usepackage{float}
\usepackage{algorithm}
\usepackage{placeins}
\usepackage{algpseudocode}
\usepackage{graphicx}
\usepackage{epstopdf}
\usepackage{scalefnt}
\usepackage[noadjust]{cite}
\usepackage[table]{xcolor}
\usepackage[cmex10]{amsmath}
\usepackage{etoolbox}
\usepackage[framemethod=TikZ]{mdframed}
\usepackage{amssymb}
\newtheorem{lem}{Lemma}

\newtheorem{pro}{Proposition}

\newtheorem{rmk}{Remark}
\newcommand*\rc{\color[rgb]{0, 0, 0}}
\newcommand{\Q}{{\vphantom{-1}}}

 \newcounter{subeqn}

\begin{document}
\title{Directional Modulation via Symbol-Level Precoding: \\ A Way to Enhance Security}
\author{
    \IEEEauthorblockN{Ashkan Kalantari, Mojtaba Soltanalian, Sina Maleki, Symeon Chatzinotas, \\ and Bj\"{o}rn Ottersten,~\IEEEmembership{Fellow,~IEEE}}
    \thanks{\footnotesize This work was supported by the National Research Fund (FNR) of Luxembourg under AFR grant  for the project ``Physical Layer Security in Satellite Communications (ref. 5798109)'', SeMIGod, and SATSENT. Ashkan Kalantari, Sina Maleki, Symeon Chatzinotas, and Bj\"{o}rn Ottersten are with the Interdisciplinary Centre for Security, Reliability and Trust (SnT), The University of Luxembourg, 4 rue Alphonse Weicker, L-2721 Luxembourg-Kirchberg, Luxembourg, ({E-mails: \{ashkan.kalantari,sina.maleki,symeon.chatzinotas,bjorn.ottersten\}@uni.lu}). M. Soltanalian is with the Department of Electrical and Computer Engineering, University of Illinois at Chicago, Chicago, IL 60607, E-mail: (msol@uic.edu). A part of this work was presented at the IEEE International Conference on Acoustics, Speech and Signal Processing (ICASSP) 2016~\cite{Kalantari:DM:ICASSP:2016}.}     
}
\maketitle
\begin{abstract}
Wireless communication provides a wide coverage at the cost of exposing information to unintended users. As an information-theoretic paradigm, secrecy rate derives bounds for secure transmission when the channel to the eavesdropper is known. However, such bounds are shown to be restrictive in practice and may require exploitation of specialized coding schemes. In this paper, we employ the concept of directional modulation and follow a signal processing approach to enhance the security of multi-user MIMO communication systems when a multi-antenna eavesdropper is present. Enhancing the security is accomplished by increasing the symbol error rate at the eavesdropper. Unlike the information-theoretic secrecy rate paradigm, we assume that the legitimate transmitter is not aware of its channel to the eavesdropper, which is a more realistic assumption. {\rc We examine the applicability of MIMO receiving algorithms at the eavesdropper.} Using the channel knowledge and the intended symbols for the users, we design security enhancing symbol-level precoders for different transmitter and eavesdropper antenna configurations. We transform each design problem to a linearly constrained quadratic program and propose two solutions, namely the iterative algorithm and one based on non-negative least squares, at each scenario for a computationally-efficient modulation. 
Simulation results verify the analysis and show that the designed precoders outperform the benchmark scheme in terms of both power efficiency and security enhancement.                 
\end{abstract}
\section{Introduction} \label{sec:intro}
\subsection{Motivation}
Wireless communications allows information flow through broadcasting; however, unintended receivers may also receive these information, with eavesdroppers amongst them. To derive a bound for secure transmission, Wyner proposed the secrecy rate concept in his seminal paper~\cite{Wyner:1975} for discrete memoryless channels. The secrecy rate defines the bound for secure transmission and proper coding is being developed to achieve this bound~\cite{Baldi:2014}. However, the secrecy rate can restrict the communication system in some aspects. Primarily, the secrecy rate requires perfect or statistical knowledge of the eavesdropper's channel state information (CSI)~\cite{Wyner:1975,Tekin:2008,Kalantari:2015:TWC,Kalantari:2015}, however, it may not be possible to acquire the perfect or statistical CSI of a passive eavesdropper in practice. In addition, in the secrecy rate approach, the transmission rate has to be lower than the achievable rate, which may conflict with the increasing rate demands in wireless communications. Furthermore, the transmit signal usually is required to follow a Gaussian distribution which is not the case in current digital communication systems. 

Recently, there has been a growing research interest on directional modulation technology and its security enhancing ability. As a pioneer,~\cite{Babakhani:2008} implements a directional modulation transmitter using parasitic antenna. This system creates the desired amplitude and phase in a specific direction by varying the length of the reflector antennas for each symbol while scrambling the symbols in other directions. The authors of~\cite{Daly:2009} suggest using a phased array at the transmitter and employ the genetic algorithm to derive the phase values of a phased array in order to create symbols in a specific direction. The directional modulation concept is later extended to directionally modulating symbols to more than one destination. In~\cite{Ding:2015:MIMO}, the singular value decomposition (SVD) is used to directionally modulate symbols in a two user system. The authors of~\cite{Yuan:2015:orth:mul:bm} derive the array weights to create two orthogonal far field patterns to directionally modulate two symbols to two different locations and~\cite{Hafez:2015} uses least-norm to derive the array weights and directionally modulate symbols towards multiple destinations in a multi-user multi-input multi-output (MIMO) system. The authors in~\cite{Kalantari:DM:ICASSP:2016} design the array weights of a directional modulation transmitter in a MIMO system to minimize the power consumption while keeping the signal-to-noise ratio (SNR) of each received signal above a specific level. The directional modulation literature focuses on practical implementation and the security enhancing characteristics of this technology. On top of the works in the directional modulation literature where antennas excitation weights change on a symbol basis, the symbol-level precoding to create constructive interference between the transmitted symbols has been developed in~\cite{cons:2015,Masouros:2015,con:glob:2015,Multicast:2016,Energy-Efficient:CI:2016} by focusing on the digital processing of the signal before being fed to the antenna array. The main difference between the directional modulation and the digital symbol-level precoding for constructive interference is that the former focuses on applying array weights in the analog domain such that the received signals on the receiving antennas have the desired amplitude and phase, whereas the latter uses symbol-level precoding for digital signal design at the transmitter to create constructive interference at the receiver. Furthermore, directional modulation was originally motivated by physical layer security, whereas symbol-level precoding by energy efficiency.
\subsection{Contributions}
\label{subsec:con}
In this paper, we study and design the optimal precoder for a directional modulation transmitter in order to enhance the security in a quasi-static fading MIMO channel where a multi-antenna eavesdropper is present. Here, enhancing the security means increasing the symbol error rate (SER) at the eavesdropper. In directional modulation, users' MIMO channel and symbols meant for the users are used to design the precoder. The precoder is designed to induce the symbols on the receiver antennas rather than generating the symbols at the transmitter and sending them, which is the case in the conventional transmit precoding~\cite{Lai-U:2004,Spencer:2004}. In other words, in the directional modulation, the modulation happens in the radio frequency (RF) level while the arrays' emitted signals pass through the wireless channel. This way, we simultaneously communicate multiple interference-free symbols to multiple users. Also, the precoder is designed such that the receiver antennas can directly recover the symbols without CSI knowledge and equalization. Therefore, assuming the eavesdropper has a different channel compared to the users, it receives scrambled symbols. In fact, the channels between the transmitter and users act as secret keys~\cite{Kui2011:sec:key} in the directional modulation. Furthermore, since the precoder depends on the symbols, the eavesdropper cannot calculate it. In contrast to the information theoretic secrecy rate paradigm, the directional modulation enhances the security by considering more practical assumptions. Particularly, directional modulation does not require the eavesdropper's CSI to enhance the security; in addition, it does not reduce the transmission rate and signals are allowed to follow a non-Gaussian distribution. In light of the above, our contributions in this paper can be summarized as follows:
\begin{enumerate}
	\item We design the optimal symbol-level precoder for a security enhancing directional modulation transmitter in a MIMO fading channel to communicate with arbitrary number of users through symbol streams. In addition, we derive the necessary condition for the existence of the precoder{\rc , which is novel compared to the digital symbol-level precoding works in~\cite{cons:2015,Masouros:2015,con:glob:2015,Multicast:2016,Energy-Efficient:CI:2016}.} The directional modulation literature mostly includes LoS analysis with one or limited number of users, and multi-user works do not {\rc design the optimal precoder to communicate symbols with arbitrary multi-antenna users from a power efficiency point of view.} 
	
	\item We {\rc analyze the applicability of various MIMO receiving algorithms at the eavesdropper. Since the imposed SER on the eavesdropper depends on the difference between the number of transmitter and the eavesdropper antennas, we consider the cases when the eavesdropper has less or more antennas than the transmitter and design a specific precoder for each case.} We minimize the transmission power for the former case and maximize the SER at the eavesdropper for the latter case to prevent or suppress successful decoding at the eavesdropper. This is done while keeping the SNR of users' received signals above a predefined threshold and thus the users' rate demands are satisfied. {\rc The analysis of different MIMO receiving algorithms at the eavesdropper and designing a precoder to maximize the SER at the eavesdropper are absent in the available directional modulation literature and digital symbol-level precoding works~\cite{cons:2015,Masouros:2015,con:glob:2015,Multicast:2016,Energy-Efficient:CI:2016}.} 
	
	\item We show that {\rc the SER imposed on the eavesdropper in the conventional precoding depends on the difference between the number of antennas of the eavesdropper and the receiver}. In our design, {\rc the SER imposed on the eavesdropper depends on the difference between the number of eavesdropper and transmitter antennas} since the precoder depends on both the channels and symbols. The transmitter, e.g., a base station, probably has more antennas than the receiver, hence, it is more likely to preserve the security in directional modulation, especially in a massive MIMO system.    
	
	\item We simplify the power and SNR minimization precoder design problems into a linearly-constrained quadratic programming problem. For faster design, we introduce new auxiliary variable to transform the constraint into equality and propose two different ways to solve the design problems. In the first way, we use the penalty method to get an unconstrained problem and solve it by proposing an iterative algorithm. Also, we prove that the algorithm converges to the optimal point. In the second one, we use the constraint to get a non-negative least squares design problem. For the latter, there are already fast techniques to solve the problem.
\end{enumerate}
\subsection{Additional Related Works to Directional Modulation}
Array switching at the symbol rate is used in~\cite{Baghdady:1990,Daly:patt:2010} to induce the desired symbols. 
In connection with~\cite{Babakhani:2008},~\cite{Lavaei:2010} studies the far field area coverage of a parasitic antenna and shows that it is a convex region. 
The technique of~\cite{Daly:2009} is implemented in~\cite{Daly:2010} using a four element microstrip patch array where symbols are directionally modulated for $\mathit{Q}$-PSK modulation. The authors of~\cite{Daly:2011} propose an iterative nonlinear optimization approach to design the array weights which minimizes the distance between the desired and the directly modulated symbols in a specific direction. 
The Fourier transform is used in~\cite{Yuan:2013:fourier,Yuan:2014:con} to create the optimal constellation pattern for $\mathit{Q}$-PSK directional modulation. In~\cite{Yuan:2014,Yuan:2014:patt:sep,Yuan:2015,Ding:2015:MIMO} directional modulation is employed along with noise injection. The authors of~\cite{Yuan:2014,Yuan:2014:patt:sep} utilize an orthogonal vector approach to derive the array weights in order to directly modulate the data and inject the artificial noise in the direction of the eavesdropper. The work of~\cite{Yuan:2014} is extended to retroactive arrays\footnote{A retroactive antenna can retransmit a reference signal back along the path which it was incident despite the presence of spatial and/or temporal variations in the propagation path.} in~\cite{Yuan:2015} for a multi-path environment. An algorithm including exhaustive search is used in~\cite{Hongzhe:2014} to adjust two-bit phase shifters for directionally modulating information. 

\subsection{Organization}
The remainder of the paper is organized as follows. In Section~\ref{sec:sys:mod}, transmitter architectures, network configuration, and the signal model are introduced. The security of the directional modulation is studied in Section~\ref{sec:dm:sec}. {\rc In Section~\ref{sec:pre:des},} the optimal precoders for the directional modulation are designed {\rc and the benchmark scheme is mentioned}. {\rc The complexity of our scheme and the benchmark method are studied in Section~\ref{sec:complex}.} In Section~\ref{sec:sim}, we present the simulation results. Finally, the conclusions are drawn in Section~\ref{sec:con}. 

\emph{Notation}: Upper-case and lower-case bold-faced letters are used to denote matrices and column vectors, respectively. The superscripts $(\cdot)^T$, $(\cdot)^*$, $(\cdot)^H$, and ${\left( \cdot   \right)^\dag }$ represent transpose, conjugate, Hermitian, and Moore-Penrose pseudo inverse operators, respectively. ${{\bf{I}}_{N \times N}}$ denotes an $N$ by $N$ identity matrix, $\rm diag(\bf{a})$ denotes a diagonal matrix where the elements of the vector $\bf{a}$ are its diagonal entries, ${\bf{a}} \circ {\bf{b}}$ is the element-wise Hadamard product, $\bf{a}_+$ denotes a vector where negative elements of the vector $\bf{a}$ are replaced by zero, $\bf{0}$ is the all zero vector, $\|\cdot\|$ is the Frobenius norm, and $|\cdot|$ represents the absolute value of a scalar. ${\mathop{\rm Re}\nolimits} \left(  \cdot  \right)$, ${\mathop{\rm Im}\nolimits} \left(  \cdot  \right)$, and $\arg \left(  \cdot  \right)$ represent the real valued part, imaginary valued part, and angle of a complex number, respectively.
\section{Signal and System Model} \label{sec:sys:mod}
\begin{figure}[]
	\centering
	\includegraphics[width=9cm]{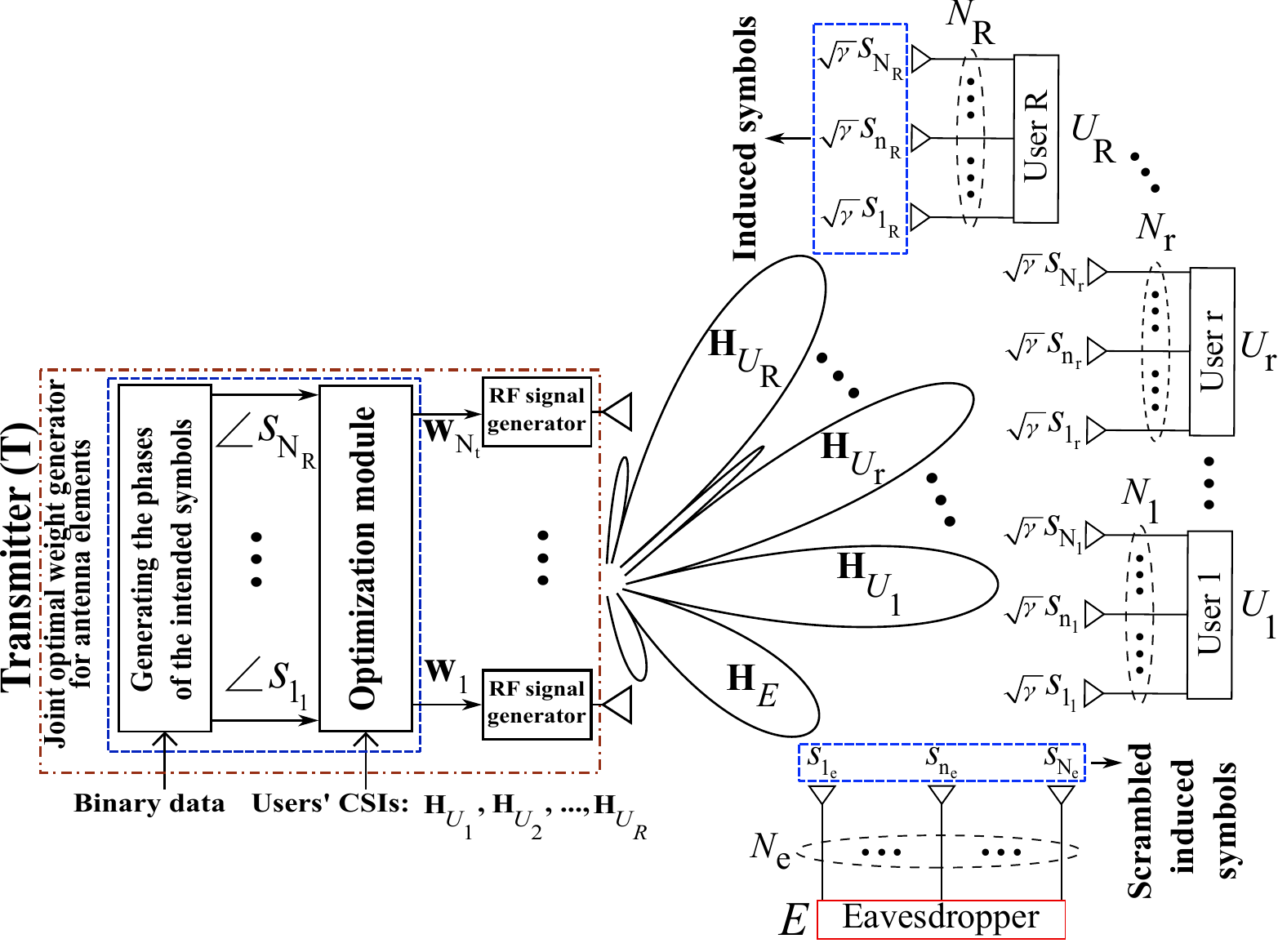}
	\caption{Generic architecture of a directional modulation transmitter, including the optimal security enhancing antenna weight generator using the proposed algorithms.}
	\label{fig:gen:sys}
\end{figure}
\begin{figure}[]
	\centering
	\includegraphics[width=6cm]{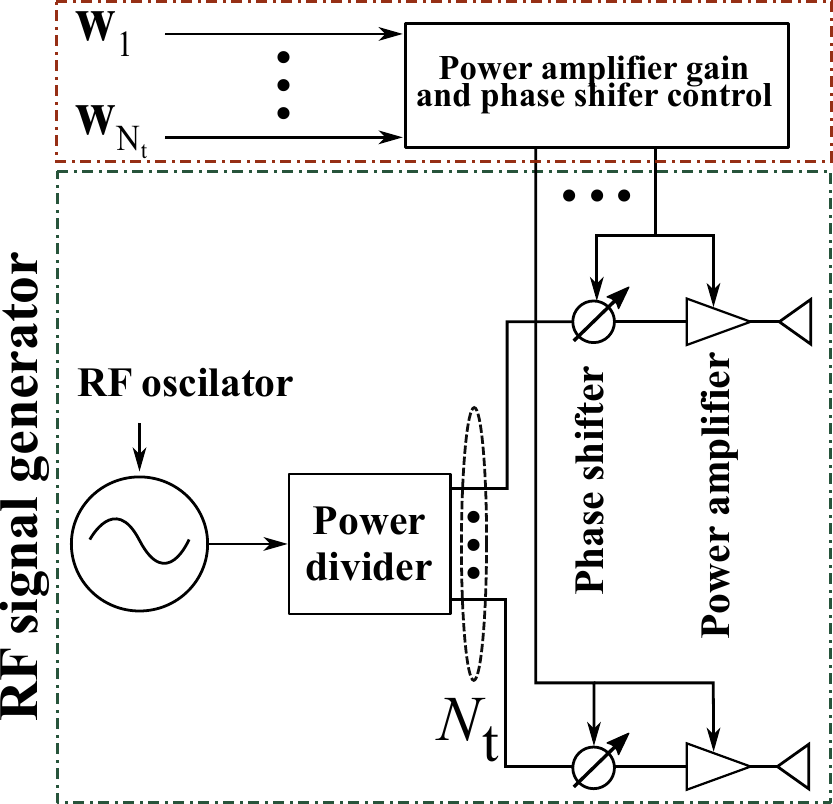}
	\caption{RF signal generation using actively driven elements, including power amplifiers and phase shifters.}
	\label{fig:sys:mod:act:ele}
\end{figure}
\begin{figure}[]
	\centering
	\includegraphics[width=6cm]{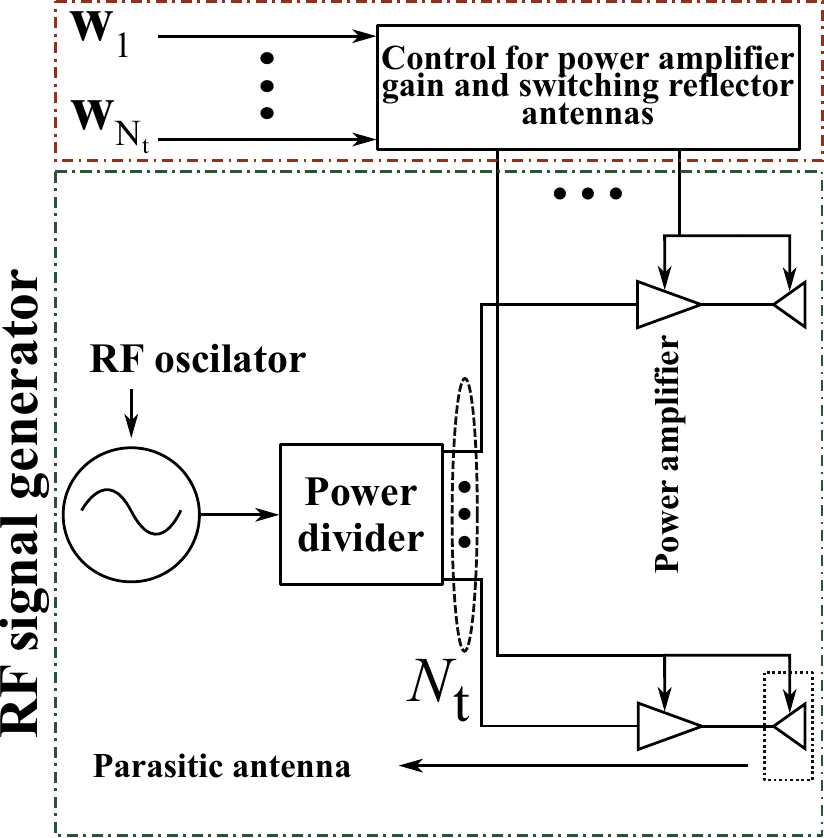}
	\caption{RF signal generation using power amplifiers and parasitic antennas.}
	\label{fig:sys:mod:par}
\end{figure}
We consider a communication network with a multi-antenna transmitter denoted by $T$, $R$ multi-antenna users denoted by $U_r$ for $r=1,...,R$ where the $r$-th user has $N_r$ antennas, and a multi-antenna eavesdropper\footnote{The same system model and solution holds for multiple colluding single-antenna eavesdroppers.} denoted by $E$ with $N_e$ antennas, as shown in Fig.~\ref{fig:gen:sys}. In addition, all the communication channels are considered to be quasi-static block fading. Two possible architectures for the RF signal generator block of Fig.~\ref{fig:gen:sys} are presented in Figures.~\ref{fig:sys:mod:act:ele} and~\ref{fig:sys:mod:par}. In Fig.~\ref{fig:sys:mod:act:ele}, power amplifiers and phase shifters are used in each RF chain to adjust the gain and the phase of the transmitted signal from each antenna. In Fig.~\ref{fig:sys:mod:par}, we adapt the technique of~\cite{Babakhani:2008} to adjust the phase using parasitic antennas in each RF chain. A parasitic antenna is comprised of a dipole antenna and multiple reflector antennas. Near field interactions between the dipole and reflector antennas creates the desired amplitude and phase in the far filed, which can be adjusted by switching the proper MOSFETs. When using parasitic antennas, the channel from each parasitic antenna to the far field needs to be LoS, and we need to acquire the CSI of the fading channel from the far field of each parasitic antenna to the receiving antennas. For simplicity, we only consider the amplitude and phase of the received signals and drop ${e^{j2\pi f t}}$, which is the carrier frequency part.  

After applying the optimal coefficients to array elements, 
the received signals by $U_r$ and $E$ are
\begin{align}
&{{\bf{y}}_{{U_r}}^\Q} = {\bf{H}}_{U_r}{\bf{w}} + {{\bf{n}}_{U_r}^\Q}, \,\, r=1,...,R
\label{eqn:eve:rec:U}
\\
&{{\bf{y}}_{E}^\Q} = {\bf{H}}_E{\bf{w}} + {{\bf{n}}_E^\Q}, 
\label{eqn:eve:rec:E}
\end{align}
where the signal ${{\bf{y}}_{U_r}^\Q}$ is an $N_r \times 1$ vector denoting the received signals by $U_r$, ${{\bf{y}}_E^\Q}$ is an $N_e \times 1$ vector denoting the received signals by $E$, ${{\bf{H}}_{{U_r}}} = {\left[ {{{\bf{h}}_{1_r}},...,{{\bf{h}}_{n_r}},...,{{\bf{h}}_{N_r}}} \right]^T}$ is an $N_r \times N_t$ matrix denoting the channel from $T$ to $U_r$, ${{\bf{h}}_{n_r}}$ is an $N_t \times 1$ vector containing the channel coefficients from the transmitter antennas to the $n$-th antenna of the $r$-th user, the channel for all users is {\rc an $N_U \times N_t$ matrix defined as} ${{\bf{H}}_U} = \left[ {{\bf{H}}_{{U_1}},...,{\bf{H}}_{{U_r}},...,{\bf{H}}_{{U_R}}} \right]^T$, ${{\bf{H}}_E}$ is an $N_e \times N_t$ matrix denoting the channel from $T$ to $E$, and $\bf{w}$ {\rc denotes the transmit precoding vector.} In directional modulation, the elements of ${{\bf{H}}_{{U_r}}}{\bf{w}} = {\left[ {\sqrt \gamma {s_{{1_r}}},...,\sqrt \gamma  {s_{{n_r}}},...,\sqrt \gamma  {s_{{N_r}}}} \right]^T}$ are the induced $\mathit{M}$-PSK symbols on the antennas of the $r$-th user, ${s_{n_r}}$ is the induced $\mathit{M}$-PSK symbol on the $n$-th antenna of the $r$-th user with instantaneous unit energy, i.e., ${\left| {s_{n_r}} \right|^2} = 1$, $\gamma$ is the SNR of the induced symbol{\rc , and $\mathit{M}$ is the $\mathit{M}$-PSK modulation order.} To detect the received symbols, $U_r$ can apply conventional detectors on each antenna. The random variables ${{\bf{n}}_{U_r}^\Q}$ and ${{\bf{n}}_E^\Q}$ denote the additive white Gaussian noise at $U_r$ and $E$, respectively. The Gaussian random variables ${{\bf{n}}_{U_r}^\Q}$ and ${{\bf{n}}_E^\Q}$ are independent and identically distributed (i.i.d.) with ${{\bf{n}}_{U_r}^\Q}\sim\mathcal{CN}({\bf{0}},  \sigma _{{n_{U_r}}}^2  {{\bf{I}}_{N_r \times N_r}} )$, and ${{\bf{n}}_E^\Q}\sim\mathcal{CN}({\bf{0}}, \sigma _{{n_E}}^2  {{\bf{I}}_{N_e \times N_e}} )$, respectively, where $\mathcal{CN}$ denotes a complex and circularly symmetric random variable. 

Throughout the paper, we assume that $T$ knows only ${{\bf{H}}_U}$ while $E$ knows both ${{\bf{H}}_U}$ and ${{\bf{H}}_E}$. In the following, we analyze the conditions under which we can enhance the system security. 
\section{Security analysis of directional modulation}   \label{sec:dm:sec}
In this section, we discuss {\rc different MIMO receiving algorithms and investigate whether $E$ can use them to estimate the received signals by the users or not}. We assume that $E$'s channel is independent from those of the users, and to consider the worst case, we assume that ${\bf{H}}_E$ is full rank. Hence, the element numbers of ${\bf{H}}_E {\bf{w}}$, i.e., received signals on $E$'s antennas, are different from those of ${\bf{H}}_{U_r}{\bf{w}}$, i.e., received signals on receiver antennas, for $r=1,...,R$. Since $\bf{w}$ depends on the symbols, $E$ cannot {\rc directly} calculate it. {\rc In the following, we analyze the capability of $E$ in using MIMO receiving algorithms to estimate $\bf{w}$.} {\rc
\subsection{Zero-Forcing Estimation}}   \label{subsec:zf}
As an approach to estimate $\bf{w}$, $E$ can remove ${\bf{H}}_E$ {\rc through zero-forcing (ZF) estimation}, and then multiply the estimated $\bf{w}$ by ${\bf{H}}_U$ to estimate the symbols. For $N_e < N_t$, $E$ cannot estimate ${\bf{H}}_{U}{\bf{w}}$ since ${\bf{H}}_E^\dag {{\bf{H}}_E} \ne {\bf{I}}$. However, when $N_e \ge N_t$, $E$ can estimate $\bf{w}$ as follows{\rc
\begin{align}
\widehat {\bf{w}} = {{\bf{G}}_1}{{\bf{y}}_E^\Q} = {\bf{w}} + {{\bf{G}}_1}{{\bf{n}}_E^\Q},
\label{eqn:eve:zf}
\end{align}
where 
\begin{align}
{{\bf{G}}_1} = {\left( {{\bf{H}}_E^H{{\bf{H}}_E}} \right)^{ - 1}}{\bf{H}}_E^H,
\label{eqn:eve:G1}
\end{align}}%
and ${\widehat {\bf{w}}}$ is the estimated $\bf{w}$ at $E$. Next, $E$ can multiply ${\widehat {\bf{w}}}$ by ${\bf{H}}_U$ to estimate the signals at receiver antennas, ${\bf{H}}_U \widehat {\bf{w}}$, as
\begin{align}
{\bf{H}}_U \widehat {\bf{w}} = {\bf{H}}_U{\bf{w}} + {\bf{H}}_U{\left( {{\bf{H}}_E^H{{\bf{H}}_E}} \right)^{ - 1}}{\bf{H}}_E^H{{\bf{n}}_E^\Q}.
\label{eqn:eve:mim}
\end{align}
Through~\eqref{eqn:eve:zf} to~\eqref{eqn:eve:mim}, $E$ virtually puts itself in the location of the users to estimate the received signal by them. {\rc The eavesdropper is capable of doing this} since we assume that it knows the users' channels, ${{\bf{H}}_U}$. This way, $E$ gets access to the secret key, which allows for observing the signals from users' point of view; however, the required process increases the noise at $E$. 
{\rc \subsection{Minimum Mean-Square Error Estimation}   \label{subsec:MMSE}
To avoid enhanced noise, $E$ can estimate $\bf{w}$ via the minimum mean-square error (MMSE) technique. The estimated symbols at $E$ through MMSE can be written as~\cite{Kay:est:I} 
\begin{align}
&{\bf{\widehat{w}}} = {{\bf{G}}_2}{{\bf{y}}_E^\Q},
\label{eqn:mmse:est}
\\
&{{\bf{H}}_U}\widehat {\bf{w}} = {{\bf{H}}_U}{{\bf{G}}_2}{{\bf{H}}_E}{\bf{w}} + {{\bf{H}}_U}{{\bf{G}}_2}{{\bf{n}}_E^\Q},
\label{eqn:mmse:est:sym}
\end{align}
with
\begin{align}
{{\bf{G}}_2} = {\left( {{\bf{H}}_E^H{\bf{C}}_{\bf{w}}^{ - 1}{{\bf{H}}_E} + {\bf{C}}_{{{\bf{N}}_E}}^{ - 1}} \right)^{ - 1}}{\bf{H}}_E^H{\bf{C}}_{\bf{w}}^{ - 1},
\label{eqn:mmse:mat}
\end{align}
where $\bf{C}_W$ is the covariance matrix of the precoding vector, $\bf{w}$, and ${{\bf{C}}_{{\bf{N}}_E}}$ is the covariance matrix of the eavesdropper noise, ${\bf{n}}_E$. As we see in~\eqref{eqn:mmse:mat}, the MMSE estimation of $\bf{w}$ at the eavesdropper requires the knowledge of ${{\bf{C}}_{\bf{W}}}$. As an approach to derive ${{\bf{C}}_{\bf{W}}}$, the eavesdropper can design $\bf{w}$ for different random sequences of $\bf{s}$ and channel realizations to derive multiple instantaneous covariance matrices as $\left( {{\bf{w}} - \overline {\bf{w}} } \right){\left( {{\bf{w}} - \overline {\bf{w}} } \right)^H}$, where ${\overline {\bf{w}} }$ is the average of $\bf{w}$. Then, $E$ can average over these instantaneous covariance matrices to calculate $\bf{C}_W$. The eavesdropper can apply the MMSE estimation approach as long as the matrix ${{\bf{H}}_E^H{\bf{C}}_{\bf{w}}^{ - 1}{{\bf{H}}_E} + {\bf{C}}_{{{\bf{N}}_E}}^{ - 1}}$ is non-singular.
\subsection{Successive Interference Cancellation and Sphere Decoding}   \label{subsec:SIC:sph:dec}
The observed signal by the eavesdropper in a conventional MIMO system is 
\begin{align}
{{\bf{y}}_E} = {{\bf{H}}_E}{\bf{W}} \bf{s} + {{\bf{n}}_E},
\label{eqn:con:MIMO}
\end{align}
where the precoding vector $\bf{W}$ depends only on the channel. The eavesdropper needs to estimate the symbol vector, $\bf{s}$, in~\eqref{eqn:con:MIMO} where its elements are drawn from a finite-alphabet set. When the successive interference cancellation (SIC) receiver is applied to a conventional MIMO receiver, each element of $\bf{s}$ is detected and reduced from the aggregated signal. This is possible since $\bf{s}$ is drawn from a finite-alphabet set~\cite{V_blast_1998}. However, in our case, the eavesdropper needs to estimate the precoding vector $\bf{w}$ whose elements take continuous values. Hence, the successive interference cancellation techniques, e.g., ZF-SIC and MMSE-SIC, cannot be applied at the eavesdropper. Furthermore, the similar argument can be followed for the sphere decoding technique~\cite{sphere:1999}, which is based on creating a sphere around the received symbol and finding the closet member of the finite-alphabet set to it.

{\rc Note that $E$ needs to estimate $\bf{w}$ whether it wants to estimate the symbols of a specific user or all the users.}

We will see in Section~\ref{sec:sim} that as the difference between $N_t$ ad $N_e$ goes higher, the imposed SER at $E$ for both ZF and MMSE estimators increases.}
\begin{rmk}
Using a large-scale array transmitter, it is more probable {\rc to have a higher difference between $N_t$ and $N_e$}. Hence, the directional modulation technique seems to be a good candidate to enhance the security when the transmitter is equipped with a large-scale array. \hfill $\blacksquare$ \hfill
\end{rmk}
{\rc
\subsection{Brute-force and maximum likelihood Approach}   
\label{subsec:bru:for}
Apart from the previous estimation approaches, the eavesdropper can follow the brute-force approach and consider all the possible symbol combinations. For a specific modulation order and total number of users' antenna, the symbol vector, $\bf{s}$, has $\mathit{M}^{N_U}$ different possibilities. This means that the eavesdropper needs to solve the design problems~\eqref{eqn:opt:pow:1},~\eqref{eqn:opt:relax:2}, or~\eqref{eqn:opt:Hw:1} $\mathit{M}^{N_U}$ times to make a look up table. Furthermore, note that the eavesdropper needs to recalculate the entire look up table if any element in ${\bf{H}}_U$ or ${\bf{H}}_E$ changes. Depending on the coherence time of the channel, this increases the computational complexity at the eavesdropper. If we assume the ideal case without noise, the eavesdropper needs to search in its look up table for ${\bf{y}}_E^{}$ to find the corresponding vector $\bf{w}$. 

Nevertheless, we have noise in practice. This requires $E$ to compare its received signal with all the computed $\mathit{M}^{N_U}$ possible cases of ${\bf{y}}_E^{}$ to find the corresponding precoding vector $\bf{w}$. As we see, the possibilities increase exponentially with $M$ and $N_U$. If we show the calculated possible cases of ${\bf{w}}$ as the set $w = \left\{ {{{\bf{w}}_1},...,{{\bf{w}}_{{M^{{N_U}}}^{}}}} \right\}$ where the cardinality of $w$ is $M^{N_U}$, the eavesdropper can follow the maximum likelihood approach to find $\bf{w}$ as 
\begin{align}
{\hat{\bf{w}}} = \arg \mathop {\min }\limits_{{{\bf{w}}_i} \in w} \,\,\,\left\| {{{\bf{y}}_E} - {{\bf{H}}_E}{{\bf{w}}_i}} \right\|_2,
\label{eqn:bru:for}
\end{align}
where ${\hat{\bf{w}}}$ is the brute-force solution. The complexity of calculating the norm of the difference of two vectors with the length $N_e$ is
\begin{align}
c_{norm} &= 2{{\bf{N}}_e}O\left( n \right) + {{\bf{N}}_e}\left( {2O\left( {{n^{1.465}}} \right) + O\left( n \right)} \right) + {{\bf{N}}_e}O\left( n \right)
\nonumber\\
&= 4 {{\bf{N}}_e}O(n) + {2{{\bf{N}}_e}O\left( {{n^{1.465}}} \right)}.
\label{eqn:norm:com}
\end{align}
Considering that the eavesdropper needs to try all the elements of the set $w$, the total complexity of the brute-force approach is given by $c_{brute-froce}=M^{N_U} (c_{norm} + c_{design})$, where $c_{design}$ is the complexity of solving~\eqref{eqn:opt:itr:u},~\eqref{eqn:opt:relax:4}, or~\eqref{eqn:opt:Hw:7}, which is quantitatively mentioned in~\eqref{eqn:com:comp:ip} and~\eqref{eqn:com:comp:fpg}. The brute-force complexity increases exponentially both in modulation order and total number of receiving antennas. To further understand the amount of computational complexity of the brute-force method, we compare it with the advanced encryption security (AES) method in the following example. For $M=32$ and $N_U=52$, the computational complexity of the brute-force method is $2^{260}(c_{norm} + c_{design})$. The complexity of the improved biclique attack to break the largest key of the AES, which has $256$ bit size, is $2^{254.27}$~\cite{biclique:attack}, which is significantly lower than the complexity of the brute-force method at the eavesdropper for the mentioned example. Computation time of the brute-force method with respect to system dimension is presented in Section~\ref{sec:sim}.

According to this section, we see that the optimal strategy at $E$ is the brute-force and maximum likelihood approaches. However, we see that this comes with an extremely large computational cost.}

\begin{rmk}
Assuming that the legitimate channel is reciprocal, the users can transmit pilots to $T$ so it can estimate ${{\bf{H}}_U}$. This way, we avoid the additional downlink channel estimation and the users do not have to send feedback bits to $T$, hence, $E$ cannot estimate ${{\bf{H}}_U}$. Assuming that $E$ knows the channel from $T$ to itself, i.e., ${{\bf{H}}_E}$, it can estimate $\bf{w}$ as in {\rc~\eqref{eqn:eve:zf} or~\eqref{eqn:mmse:est}}, but it cannot perform~\eqref{eqn:eve:mim} {\rc or~\eqref{eqn:mmse:est:sym}} to estimate the received signals on the receiver antennas.             \hfill $\blacksquare$ \hfill	
\end{rmk} 
\vspace{3pt}

In the next section, optimal symbol-level precoders for the directional modulation are designed to enhance the security.
\section{Optimal Precoder Design for Directional Modulation}
\label{sec:pre:des}
In this section, we define the underlaying problems to design the security enhancing symbol-level precoder for the directional modulation. {\rc Since the SER at $E$ depends on the difference between $N_t$ and $N_e$, we consider the cases $N_e < N_t$ and $N_e \ge N_t$ and design a specific precoder for each of them. The case $N_e < N_t$ focuses on energy efficiency, hence, we also perform relaxed phase analysis for this case.} 
\subsection{The Case of Strong Transmitter ($N_e < N_t${\rc , Fixed Phase)}} 
\label{subsec:lar:tx}
In wireless transmission, adaptive coding and modulation (ACM) is used to enhance the link performance and the channel capacity. In ACM, the transmission power, coding rate, and the modulation order is set according to the channel signal to noise ratio (SNR)~\cite{Goldsmith:1998}. Based on this, we preserve the SNR of the induced symbol on the receiver antenna above or equal to a specific level to successfully decode it. Here, we only focus on the SNR of an uncoded signal since considering SNR of a coded transmission based on ACM is beyond the scope of this paper.  

To avoid a non-convex design problem, we use the required signal properties at the receiver to formulate a convex design problem. In our design, a specific fixed phase is required for the received signal at each receiver antenna. {\rc Since the phase of the received signal at each receiving antenna, ${{\bf{h}}_{{n_r}}^T{\bf{w}}}$, is the same as the phase of the intended symbol, $s_{n_r}$}, if the required SNR{\rc, $\gamma$, of the received signal} increases, the in-phase{\rc, ${\mathop{\rm Re}\nolimits} \left( {{\bf{h}}_{{n_r}}^T{\bf{w}}} \right)$,} and quadrature-phase{\rc, ${\mathop{\rm Im}\nolimits} \left( {{\bf{h}}_{{n_r}}^T{\bf{w}}} \right)$,} parts {\rc will increase} in the same proportion {\rc to satisfy the required SNR}. Since the received signal by each antenna is complex {\rc valued}, we separately consider amplitudes of the in-phase and quadrature-phase parts of {\rc the received signal} on the receiver antenna instead of {\rc its power}. If we show the real and imaginary valued parts of $s_{n_r}$ as ${\rm{Re}}\left( {{s_{n_r}}} \right)$ and ${\rm{Im}}\left( {{s_{n_r}}} \right)$, the required in-phase and quadrature-phase thresholds {\rc of the received signal} are defined as 
\begin{align}
\sqrt \gamma  {\mathop{\rm Re}\nolimits} \left( {{s_{n_r}}} \right), \,\,   \sqrt \gamma  {\mathop{\rm Im}\nolimits} \left( {{s_{n_r}}} \right).
\label{eqn:rl:im}
\end{align}
Since ${\left| {{s_{{n_r}}}} \right|^2} = 1$, we can see that ${\gamma } = {\gamma }{\mathop{\rm Re^2}\nolimits} \left( {{s_{n_r}}} \right) + {\gamma }{\mathop{\rm Im^2}\nolimits} \left( {{s_{n_r}}} \right)${\rc , which satisfies the SNR constraint.}

We design the directional modulation precoder to minimize the total transmit power such that 1) the signals received by the $n$-th antenna of the $r$-th user result in a phase equal to that of $s_{n_r}$, and 2) the signals received by the $n$-th antenna of the $r$-th user create in-phase and quadrature-phase signal levels satisfying the thresholds defined in~\eqref{eqn:rl:im}. Accordingly, the precoder design problem is defined as
\begin{subequations}
\begin{align}
&\mathop {\min }\limits_{\bf{w}} \,\, {\left\| {\bf{w}} \right\|^2}
\nonumber\\
& \,\, \text{s.t.}   \,\,\,\,     \arg \left( {  {\bf{h}}_{n_r}^T  {\bf{w}}} \right) = \arg \left( {s_{n_r}} \right),                           
\label{subeq:w:c11}
\\
&             \qquad \,\, {\rm{Re}}\left( {  {\bf{h}}_{n_r}^T  {\bf{w}}  } \right) \ge \sqrt \gamma  {\rm{R}}{{\rm{e}}}\left( {s_{n_r}} \right), 
\label{subeq:w:c21}
\end{align}
\label{eqn:opt:pow:1}%
\end{subequations}
for $r = 1,...,R$ and $n=1,...,N$. 
{\rc Since the phase of the induced symbol is fixed, we just need to put the signal level constraint over the real or imaginary part of the received signal on each receiving antenna. Hence, we have included the constraint over the value of the real part in~\eqref{subeq:w:c21}.} {\rc Generally, some constraints of~\eqref{eqn:opt:pow:1} are satisfied with inequality and the rest are satisfied with equality~\cite{Sidiropoulos:2006}. This depends on the difference between $N_t$ and $N_U$. We will also show this through simulations in Section~\ref{sec:sim}. In the case that each user is associate with a precoder, i.e., the transmitter designs ${\bf{w}}_1$,..., ${\bf{w}}_K$ for $K$ users, the constraints are satisfied with equality at the optimal point~\cite{mu:user:beam:2014}.} If both sides of~\eqref{subeq:w:c21} are negative, the signal level constraints may not be satisfied. Since~\eqref{subeq:w:c11} holds at the optimal point, ${\mathop{\rm Re}\nolimits} \left( {{\bf{h}}_{{n_r}}^T{\bf{w}}} \right)$ has the same sign as ${\mathop{\rm Re}\nolimits} \left( {{s_{{n_r}}}} \right)$ at the optimal point. Therefore, we can multiply both sides of~\eqref{subeq:w:c21} by ${\rm Re}(s_{n_r})$ to get
\begin{subequations}
	\begin{align}
	&\mathop {\min }\limits_{\bf{w}} \,\, {\left\| {\bf{w}} \right\|^2}
	\nonumber\\
	& \,\, \text{s.t.}   \,\,\,\,     \arg \left( {  {\bf{h}}_{n_r}^T  {\bf{w}}} \right) = \arg \left( {s_{n_r}} \right),                           
	\label{subeq:w:c12}
	\\
	&             \qquad \,\,  {\rm{Re}}\left( {s_{n_r}} \right){\rm{Re}}\left( {  {\bf{h}}_{n_r}^T  {\bf{w}}  } \right) \ge \sqrt \gamma  {\rm{R}}{{\rm{e}}^2}\left( {s_{n_r}} \right). 
	\label{subeq:w:c22}
	\end{align}
	\label{eqn:opt:pow:2}%
\end{subequations}
To simplify~\eqref{eqn:opt:pow:2}, we can rewrite the phase constraint in~\eqref{subeq:w:c12} as 
\begin{align}
{\mathop{\rm Re}\nolimits} \left( {{\bf{h}}_{n_r}^T{\bf{w}}} \right){\alpha _{n_r}} - {\mathop{\rm Im}\nolimits} \left( {{\bf{h}}_{n_r}^T{\bf{w}}} \right) = 0, \,\,\,\, \forall \, n, \,\, \forall \, r, 
\label{eqn:pha}
\end{align}
where ${\alpha_{n_r}} = \tan \left( {{s_{n_r}}} \right)$. Since $\rm tan(\cdot)$ repeats after a $\pi$ radian period\footnote{If the phase of the $\mathit{M}$-PSK constellation falls on the points where $\tan$ function is undefined, e.g., $\frac{\pi }{2}$, we can add phase offset to the modulation.}, symbols with different phases can have the same $\rm tan$ value, e.g., $\tan \left( {\frac{\pi }{4}} \right) = \tan \left( {\frac{{3\pi }}{4}} \right)$. Therefore, replacing~\eqref{subeq:w:c12} with~\eqref{eqn:pha} creates ambiguity. To avoid this, we can add the constraint 
\begin{align}
&{\rm{Re}}\left( {{s_{n_r}}} \right){\rm{Re}}\left( {{\bf{h}}_{n_r}^T{\bf{w}}} \right) \ge 0, 
\label{eqn:amb:con}
\end{align}
to the design problem~\eqref{eqn:opt:pow:2} to avoid ambiguity. 
Interestingly, {\rc constraint~\eqref{eqn:amb:con} is} already present {\rc in~\eqref{subeq:w:c22}}. Note that~\eqref{eqn:pha} and~\eqref{eqn:amb:con} together are equivalent to~\eqref{subeq:w:c11}, so the required conditions to go from~\eqref{eqn:opt:pow:1} to~\eqref{eqn:opt:pow:2} still hold. Putting together the constraints~\eqref{eqn:pha} and~\eqref{subeq:w:c22} for all the users,~\eqref{eqn:opt:pow:2} is written into the following compact form
\begin{subequations}
\begin{align}
&\mathop {\min }\limits_{\bf{w}} \,\, {\left\| {\bf{w}} \right\|^2}
\nonumber\\
& \,\, \text{s.t.}   \,\,\,\,\,     {\bf{A}}{\rm{Re}}\left( {{{\bf{H}}_U}{\bf{w}}} \right) - {\rm{Im}}\left( {{{\bf{H}}_U}{\bf{w}}} \right) = {\bf{0}},    
\\
&                    \qquad \,\,  {\mathop{\rm Re}\nolimits} \left( {\bf{S}} \right){\mathop{\rm Re}\nolimits} \left( {{{\bf{H}}_U}{\bf{w}}} \right) \ge \sqrt {\gamma} \,  {{\bf{s}}_{r}}, 
\label{subeq:w:c13}
\end{align}
\label{eqn:opt:pow:3}%
\end{subequations}
where ${\bf{S}} = \rm diag\left( {\bf{s}} \right)$, ${\bf{s}}$ is an ${N_{U} \times 1}$ vector containing all the intended $\mathit{M}$-PSK symbols for the users with ${N_{{U}}} = \sum\nolimits_{r = 1}^R {{N_r}}$, ${{\bf{s}}_r} = {\mathop{\rm Re}\nolimits} \left( {\bf{s}} \right) \circ {\mathop{\rm Re}\nolimits} \left( {\bf{s}} \right)$, 
${\bf{A}} = \rm diag\left( \boldsymbol{\alpha} \right)$, $\boldsymbol{\alpha}  = {\left[ {{\alpha_{1_1}},...,{\alpha _{n_r}},...,{\alpha _{N_R}}} \right]^T}$.

To remove the real and imaginary valued parts from~\eqref{eqn:opt:pow:3}, we can use ${{\bf{H}}_U} = {\mathop{\rm Re}\nolimits} \left( {{{\bf{H}}_U}} \right) + i{\mathop{\rm Im}\nolimits} \left( {{{\bf{H}}_U}} \right)$ and ${\bf{w}} = {\mathop{\rm Re}\nolimits} \left( {\bf{w}} \right) + i{\mathop{\rm Im}\nolimits} \left( {\bf{w}} \right)$ presentations to separate the real and imaginary valued components of ${\bf{H}}_U{\bf{w}}$ as 
\begin{align}
{{\bf{H}}_U}{\bf{w}} =& {\rm{Re}}\left( {{{\bf{H}}_U}} \right){\rm{Re}}\left( {\bf{w}} \right) - {\rm{Im}}\left( {{{\bf{H}}_U}} \right){\rm{Im}}\left( {\bf{w}} \right)
\nonumber\\
&+ i\left[ {{\rm{Re}}\left( {{{\bf{H}}_U}} \right){\rm{Im}}\left( {\bf{w}} \right) + {\rm{Im}}\left( {{{\bf{H}}_U}} \right){\rm{Re}}\left( {\bf{w}} \right)} \right],
\label{eqn:Hw}
\end{align}
which leads into the following expressions
\begin{align}
{\rm{Re}}\left( {{\bf{H}}_U{\bf{w}}} \right) = {\bf{H}}_{U_1}\widetilde {\bf{w}}, \,\, {\rm{Im}}\left( {{\bf{H}}_U{\bf{w}}} \right) = {\bf{H}}_{U_2}\widetilde {\bf{w}},
\label{eqn:rel:ima:Hw}
\end{align}
where $\widetilde {\bf{w}} = {\left[ {{\rm{Re}}\left( {\bf{w}}^T \right),{\rm{Im}}\left( {\bf{w}}^T \right)} \right]^T}$, ${{\bf{H}}_{{U_1}}} = \left[ {{\mathop{\rm Re}\nolimits} \left( {{{\bf{H}}_U}} \right), - \rm Im\left( {{{\bf{H}}_U}} \right)} \right]$, and ${{\bf{H}}_{{U_2}}} = \left[ {\rm Im\left( {{{\bf{H}}_U}} \right),{\mathop{\rm Re}\nolimits} \left( {{{\bf{H}}_U}} \right)} \right]$. Also, it is easy to see that ${\left\| {\widetilde {\bf{w}}} \right\|^2} = {\left\| {\bf{w}} \right\|^2}$. 

Using the equivalents of ${\rm{Re}}\left( {{\bf{H}}_U{\bf{w}}} \right)$ and ${\rm{Im}}\left( {{\bf{H}}_U{\bf{w}}} \right)$ derived in~\eqref{eqn:rel:ima:Hw},~\eqref{eqn:opt:pow:3} transforms into
\begin{subequations}
\begin{align}
& \mathop {\min }\limits_{\widetilde {\bf{w}}}  \,\,\,\,  \left\| {\widetilde {\bf{w}}} \right\|^2
\nonumber\\
& \, \text{s.t.}   \,\,\,\,\,\,\,  \left( {{\bf{A}}{{\bf{H}}_{{U_1}}} - {{\bf{H}}_{{U_2}}}} \right)\widetilde {\bf{w}} =\bf{0},
\label{subeq:w:c41}
\\
&                  \qquad \,\,\,\,   {\mathop{\rm Re}\nolimits} \left( {\bf{S}} \right){{\bf{H}}_{{U_1}}}\widetilde {\bf{w}} \ge \sqrt{\gamma} \,  {{\bf{s}}_r}.
\label{subeq:w:c42}
\end{align}
\label{eqn:opt:pow:4}%
\end{subequations}
\begin{pro}
A necessary condition for the existence of the optimal precoder for the directional modulation is $N_t > \frac{r^{'}}{2}$ where $r^{'}$ is the rank of ${{\bf{A}}{{\bf{H}}_{{U_1}}} - {{\bf{H}}_{{U_2}}}}$. If ${{\bf{A}}{{\bf{H}}_{{U_1}}} - {{\bf{H}}_{{U_2}}}}$ is full rank, the necessary condition becomes $N_t > \frac{N_{U}}{2}$, which means that the number of transmit antennas needs to be more than half of the total number of receiver antennas. 
\label{pro:nec:con}
\end{pro}
\begin{proof}
Constraint~\eqref{subeq:w:c41} shows that $\widetilde {\bf{w}}$ should lie in the null space of the matrix ${{\bf{A}}{{\bf{H}}_{{U_1}}} - {{\bf{H}}_{{U_2}}}}$. If the SVD of ${{\bf{A}}{{\bf{H}}_{{U_1}}} - {{\bf{H}}_{{U_2}}}}$ is shown by ${\bf{U}}\Sigma {{\bf{V}}^H}$, the orthonormal basis for the null space of ${{\bf{A}}{{\bf{H}}_{{U_1}}} - {{\bf{H}}_{{U_2}}}}$ are the last $2 N_t - r^{'}$ columns of the matrix $\bf{V}$ with $r^{'}$ being the rank of ${{\bf{A}}{{\bf{H}}_{{U_1}}} - {{\bf{H}}_{{U_2}}}}$~\cite{Strang2009}. If ${{\bf{A}}{{\bf{H}}_{{U_1}}} - {{\bf{H}}_{{U_2}}}}$ is full rank, we have $r^{'}=N_{U}$. For~\eqref{eqn:opt:pow:4} to be feasible, the mentioned null space should exist, meaning that $2 N_t - r^{'} >0$.
\end{proof}
\vspace{3pt}
Provided that the necessary condition of Proposition~\ref{pro:nec:con} is met, a sufficient condition can be proposed from a geometrical point of view; namely that the feasible set of~\eqref{eqn:opt:pow:4} is not empty. This holds if and only if the intersection of the linear spaces in the constraint set constitutes a non-empty set. 	

According to Proposition~\ref{pro:nec:con}, the null space of ${{\bf{A}}{{\bf{H}}_{{U_1}}} - {{\bf{H}}_{{U_2}}}}$ spans $\widetilde {\bf{w}}$ as $\widetilde {\bf{w}} = {\bf{E \boldsymbol \lambda }}$ where 
\begin{align}
{\bf{E}} = \left[ {{{\bf{v}}_{r^{'} + 1}},...,{{\bf{v}}_{2 N_t}}} \right] , \, \, \boldsymbol \lambda  = \left[ {{\lambda_1},...,{\lambda _{2N_t-r^{'}}}} \right].
\label{eqn:E:lam}
\end{align}
By replacing $\widetilde {\bf{w}}$ with $\bf{E \boldsymbol \lambda }$,~\eqref{eqn:opt:pow:4} boils down into
\begin{align}
& \mathop {\min }\limits_{ {\boldsymbol{\lambda}}}  \,\,\,\,  \left\| { {{\boldsymbol{\lambda }}}} \right\|^2
\nonumber\\
& \, \text{s.t.}   \,\,\,\,\,\,\,\,  {\mathop{\rm Re}\nolimits} \left( {\bf{S}} \right){{\bf{H}}_{{U_1}}} {\bf{E}}  {\boldsymbol{\lambda}} \ge \sqrt {\gamma} \,  {{\bf{s}}_r},
\label{eqn:opt:pow:5}
\end{align}
Problem\footnote{The design problem~\eqref{eqn:opt:pow:5} can be extended to M-QAM modulation~\cite{con:glob:2015} by changing the constraint into equality. A detailed derivation falls beyond the scope of this paper.}~\eqref{eqn:opt:pow:5} is a convex linearly constrained quadratic programming problem and can be solved efficiently using standard convex optimization techniques. The design problem~\eqref{eqn:opt:pow:5} needs to be solved once for each set of the symbols, ${\bf{s}}_T$. Using optimization packages such as CVX to solve~\eqref{eqn:opt:pow:5} can be time consuming, hence, we propose two other approaches to solve~\eqref{eqn:opt:pow:5}. 
\vspace{3pt}
\subsubsection{Iterative solution} 
\label{subsubsec:lar:tx:sol:a}
In this part, we propose an iterative approach to solve~\eqref{eqn:opt:pow:5}. To do so, first, we define a real valued auxiliary vector denoted by $\bf{u}$ to change the inequality constraint of~\eqref{eqn:opt:pow:5} into equality as  
\begin{align}
& \mathop {\min }\limits_{ {\boldsymbol{\lambda}}, \bf{u}}  \,\,\,\,  \left\| { {{\boldsymbol{\lambda }}}} \right\|^2
\nonumber\\
& \, \text{s.t.}   \,\,\,\,\,\,\,\,  {{\bf{B} \boldsymbol{\lambda} }} = \sqrt \gamma {\mkern 1mu} {{\bf{s}}_r}+{\bf{u}}, \,\, \bf{u} \ge 0.
\label{eqn:opt:pow:7}
\end{align}
{\rc where ${\bf{B}} = {\rm{Re}}\left( {\bf{S}} \right){{\bf{H}}_{{U_1}}}{\bf{E}}$}. Using the penalty method~\cite{Boyd}, we can write~\eqref{eqn:opt:pow:7} as an unconstrained optimization problem
\begin{align}
& \mathop {\min }\limits_{\boldsymbol{\lambda} ,{\bf{u}} \ge 0} \,\,\,{\left\| \boldsymbol{\lambda}  \right\|^2} + \eta {\left\| {{\bf{B}}\boldsymbol{\lambda}  - \left( {\sqrt \gamma {\rc {{\bf{s}}_r}} + {\bf{u}}} \right)} \right\|^2},
\label{eqn:opt:pow:pen}
\end{align}
which is equivalent to~\eqref{eqn:opt:pow:7} when $\eta \to \infty$. We can solve~\eqref{eqn:opt:pow:pen} using an iterative approach by first optimizing $\bf{u}$ and considering $\boldsymbol \lambda$ to be fixed, and then optimizing $\bf{u}$ and considering $\boldsymbol{\lambda}$ to be fixed. In the following, we mention these two optimization problems and their closed-form solutions. 

When optimizing over $\bf{u}$ and keeping $\boldsymbol{\lambda}$ fixed, the optimization problem to be solved can be written as
\begin{align}
& \mathop {\min }\limits_{{\bf{u}} \ge 0} \,\,\,{\left\| {{\bf{u}} - \left( {{\bf{B}} \boldsymbol{\lambda}  - \sqrt \gamma {\rc {{\bf{s}}_r}}} \right)} \right\|^2}.
\label{eqn:opt:itr:u}
\end{align}
The closed-form solution of~\eqref{eqn:opt:itr:u} is given in Lemma~\ref{lem:itr:u}.
\begin{lem}
The closed-form solution of~\eqref{eqn:opt:itr:u} is ${\bf{u^{\star}}} = {\left( {{\bf{B}} \boldsymbol{\lambda}  - \sqrt \gamma {\rc {{\bf{s}}_r}}} \right)_ + }$.
\label{lem:itr:u}
\end{lem}
\begin{proof}
To solve~\eqref{eqn:opt:itr:u}, we need to minimize the distance between the vectors $\bf{u}$ and ${\left( {{\bf{B}} \boldsymbol{\lambda}  - \sqrt \gamma {\rc {{\bf{s}}_r}}} \right)}$. Since $\boldsymbol{\lambda}$ is fixed, the elements of ${\left( {{\bf{B}} \boldsymbol{\lambda}  - \sqrt \gamma {\rc {{\bf{s}}_r}}} \right)}$ are known. If an element of ${{\bf{B \boldsymbol{\lambda} }} - \sqrt \gamma {\rc {{\bf{s}}_r}}}$ is nonnegative, we pick up the same value for the corresponding element of $\bf{u}$. If an element of ${{\bf{B \boldsymbol{\lambda} }} - \sqrt \gamma {\rc {{\bf{s}}_r}}}$ is negative, we pick up zero for the corresponding element of $\bf{u}$ since $\bf{u} \ge 0$. This is equivalent to picking up $\bf{u}$ as 
\begin{align}
& {\bf{u^{\star}}} = {\left( {{\bf{B}} \boldsymbol{\lambda}  - \sqrt \gamma {\rc {{\bf{s}}_r}}} \right)_ + }.
\label{eqn:opt:itr:u:cf}
\end{align}
\end{proof}
\vspace{3pt}
When optimizing over $\boldsymbol{\lambda}$ and keeping $\bf{u}$ fixed, the optimization problem is 
\begin{align}
& \mathop {\min }\limits_{\boldsymbol{\lambda} } \,\,\,{\left\| \boldsymbol{\lambda}  \right\|^2} + \eta {\left\| {{\bf{B}}\boldsymbol{\lambda}  - \left( {\sqrt \gamma {\rc {{\bf{s}}_r}} + {\bf{u}}} \right)} \right\|^2}.
\label{eqn:opt:itr:lam}
\end{align}
The closed-form solution of~\eqref{eqn:opt:itr:lam} is given in Lemma~\ref{lem:itr:lam}.
\begin{lem}
	The closed-form solution of~\eqref{eqn:opt:itr:lam} is ${\boldsymbol{\lambda} }^\star  = {\left( {\frac{{\bf{I}}}{\eta } + {{\bf{B}}^T}{\bf{B}}} \right)^{ - 1}}{{\bf{B}}^T}\left( {{\bf{a}} + {\bf{u}}} \right)$.
	\label{lem:itr:lam}
\end{lem}
\begin{proof}
First, we expand~\eqref{eqn:opt:itr:lam} as
\begin{align}
 f\left( {\boldsymbol{\lambda} }  \right) =& {\left\| {\boldsymbol{\lambda} }  \right\|^2} + \eta {\left\| {{\bf{B}} {\boldsymbol{\lambda} }  - \left( {{\rc \sqrt{ \gamma}} {\rc {{\bf{s}}_r}} + {\bf{u}}} \right)} \right\|^2}
\nonumber\\
 =& {{\boldsymbol{\lambda} } ^T}\left( {{\bf{I}} + \eta {{\bf{B}}^T}{\bf{B}}} \right){\boldsymbol{\lambda} }  - 2\eta {{\boldsymbol{\lambda} } ^T}\left( {{{\bf{B}}^T}\gamma {\rc {{\bf{s}}_r}} + {{\bf{B}}^T}{\bf{u}}} \right) 
\nonumber\\
&+ \eta {\left( {\sqrt \gamma  {\rc {{\bf{s}}_r}} + {\bf{u}}} \right)^T}\left( {\sqrt \gamma  {\rc {{\bf{s}}_r}} + {\bf{u}}} \right).
\label{eqn:opt:itr:lam:ex}
\end{align}
Taking the derivative of $f\left( {\boldsymbol{\lambda} }  \right)$ with respect to ${\boldsymbol{\lambda} }$ yields
\begin{align}
{\boldsymbol{\lambda} }^\star  = {\left( {\frac{{\bf{I}}}{\eta } + {{\bf{B}}^T}{\bf{B}}} \right)^{ - 1}}{{\bf{B}}^T}\left( {{\bf{a}} + {\bf{u}}} \right).
\label{eqn:opt:itr:lam:cf}
\end{align}
Since ${{\bf{B}}^T}{\bf{B}}$ is positive semidefinite, addition of $\frac{{\bf{I}}}{\eta }$ to ${{\bf{B}}^T}{\bf{B}}$ for $\eta \neq \infty$ leads into diagonal loading of ${{\bf{B}}^T}{\bf{B}}$, which makes ${\frac{{\bf{I}}}{\eta } + {{\bf{B}}^T}{\bf{B}}}$ invertible.
\end{proof}
\vspace{3pt}
Using the closed-form solutions mentioned in Lemmas~\ref{lem:itr:u} and~\ref{lem:itr:lam}, we propose Algorithm~\ref{alg:itr} to solve~\eqref{eqn:opt:pow:pen}, where the matrix inversion in~\eqref{eqn:opt:itr:lam:cf} needs to be calculated once per symbol transmission.
\begin{lem}
Algorithm~\ref{alg:itr} monotonically converges to the optimal point.
\end{lem}
\begin{proof}
Let's denote the objective function in~\eqref{eqn:opt:pow:pen} by $f\left( {{\boldsymbol {\lambda}},{\bf{u}}} \right)$. Assume ${{{\boldsymbol{\lambda}} _0}}$ and ${\bf{u}}_0$ are initial values of $f\left( {{\boldsymbol {\lambda}},{\bf{u}}} \right)$. Using ${{{\boldsymbol{\lambda}} _0}}$ in Algorithm~\ref{alg:itr} gives us ${{\bf{u}}^\star}$ and ${{{\boldsymbol{\lambda}}^\star}}$ from~\eqref{eqn:opt:itr:u:cf} and~\eqref{eqn:opt:itr:lam:cf}, respectively, which results in
\begin{align}
f\left( {{\boldsymbol{\lambda} ^\star},{{\bf{u}}^\star}} \right) \le f\left( {{\boldsymbol{\lambda}_0},{{\bf{u}}^\star}} \right) \le f\left( {{\boldsymbol{\lambda}_0},{{\bf{u}}_0}} \right).
\label{eqn:opt:seq}
\end{align}
Since fixing ${\boldsymbol {\lambda}}$,~\eqref{eqn:opt:itr:u}, or $\bf{u}$,~\eqref{eqn:opt:itr:lam}, leads into a convex function, each iteration in Algorithm~\ref{alg:itr} monotonically gets closer to the optimal point. This along with the fact that $f\left( {{\boldsymbol {\lambda}},{\bf{u}}} \right)$ is lower bounded at zero,  guarantees the convergence of Algorithm~\ref{alg:itr} to the optimal point.
\end{proof}	
\vspace{3pt}    
\algnewcommand{\algorithmicgoto}{\textbf{Go to}}%
\algnewcommand{\Goto}[1]{\algorithmicgoto~\ref{#1}}%
\begin{algorithm}[t]
	\caption{Iterative approach to solve~\eqref{eqn:opt:pow:pen}}
	\begin{algorithmic}[1]
		\State Pick up ${\boldsymbol{\lambda}_n } \in {\rm I\!R}^{2 N_t} $ and $\eta \in (0,\left. \infty  \right]$;
		\State Substitute ${\boldsymbol{\lambda}_n }$ in~\eqref{eqn:opt:itr:u:cf} to get ${\bf{u}}_n$; \label{1}
		\State Substitute ${\bf{u}}_n$ in~\eqref{eqn:opt:itr:lam:cf} to get ${\boldsymbol{\lambda}_{n+1} }$;
		\If{ $\left\| {{\boldsymbol{\lambda}_n } - {\boldsymbol{\lambda}_{n+1} }} \right\| \ge \epsilon$ }
		\State $n=n+1$;
		\State  \Goto{1};
		\EndIf
	\end{algorithmic}
	\label{alg:itr}
\end{algorithm}
\vspace{3pt}
\subsubsection{Non-negative least squares} 
\label{subsubsec:lar:tx:sol:b}
We can derive $\boldsymbol\lambda$ using the constraint of~\eqref{eqn:opt:pow:7} as 
\begin{align}
{\boldsymbol{\lambda }} = \,{{\bf{B}}^\dag }\left( {\sqrt \gamma  {\rc {{\bf{s}}_r}} + {\bf{u}}} \right).
\label{eqn:opt:lambda}
\end{align}
Replacing the $\boldsymbol{\lambda}$ derived in~\eqref{eqn:opt:lambda} back into the objective of~\eqref{eqn:opt:pow:7} yields
\begin{align}
&\mathop {\min }\limits_{\bf{u}} {\mkern 1mu} {\mkern 1mu} {\mkern 1mu} {\mkern 1mu} {\left\| {{{\bf{B}}^\dag }{\bf{u}} + {\mkern 1mu} \sqrt \gamma  {{\bf{B}}^\dag }{\rc {{\bf{s}}_r}}} \right\|^2}
\nonumber\\
& \, \text{s.t.}   \,\,\,\,\,\,\,\, {\bf{u}} \ge 0,
\label{eqn:opt:nnls}
\end{align}
which is a non-negative least squares optimization problem. Since ${{\bf{B}}^\dag }$ and ${\sqrt \gamma  {{\bf{B}}^\dag }{\rc {{\bf{s}}_r}}}$ are real valued, we can use the method of~\cite{Solving:1995} or its fast version~\cite{Bro1997b} to solve~\eqref{eqn:opt:nnls}. 
{\rc We analyze the computational complexity of the non-negative least squares in Section~\ref{sec:complex} and mention its} computational time in Section~\ref{sec:sim}. Similar to Section~\ref{subsubsec:lar:tx:sol:a}, ${{\bf{B}}^\dag }$ needs to be calculated once per symbol transmission.  
{\rc \subsection{The Case of Strong Transmitter ($N_e < N_t${\rc , Relaxed Phase)}}       
\label{subsec:lar:tx:relax}
The phases of the received signals in~\eqref{eqn:opt:pow:1} are fixed, which decreases the degrees of freedom in designing $\bf{w}$, and consequently the power efficiency. To improve the power efficiency in the transmitter side, we can consider a region instead of a line for the phase of the received signal on each receiving antenna. In the $\mathit{M}$-PSK modulation, each symbol has a detection region within $\pm \frac{\pi }{M}$ degrees of its phase. The detection and relaxed phase regions for a reference symbol $s_0$ with the angle ${\varphi _{{s_0}}}=arg(s_0)$ are shown in Fig.~\ref{fig:sys:mod:relaxed}~\cite{Masouros:relaxed:2010}. According to the characterization in Fig.~\ref{fig:sys:mod:relaxed}, the relaxed phase design problem is defined as~\cite{Masouros:relaxed:2010,Masouros:2015,Energy-Efficient:CI:2016} 
\begin{figure}[]
	\centering
	\includegraphics[width=6cm]{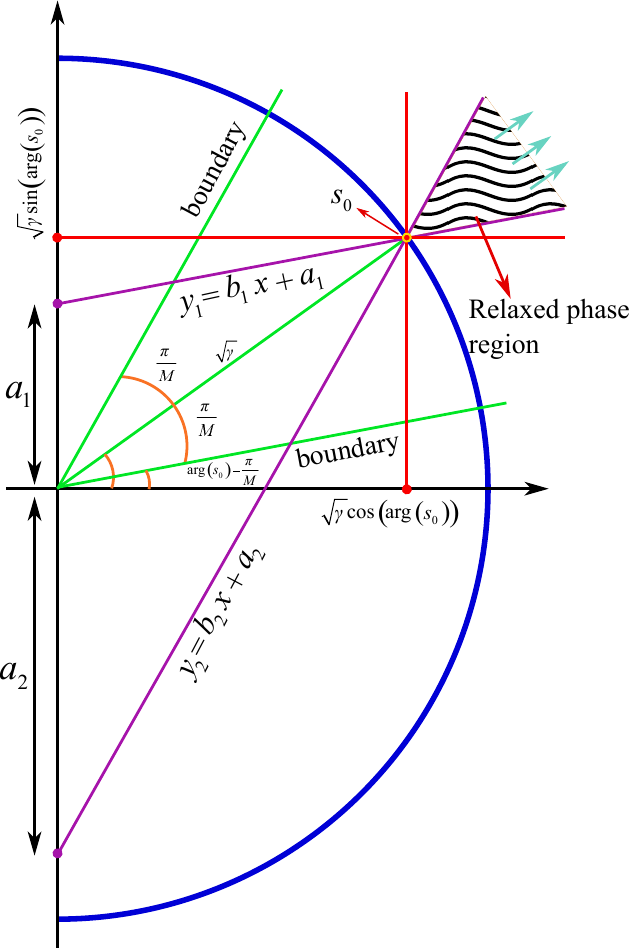}
	\caption{{\rc Relaxed phase characterization of directional modulation design for symbol $s_0$ from $\mathit{M}$-PSK modulation.}}
	\label{fig:sys:mod:relaxed}
\end{figure}
\begin{subequations}
	\begin{align}
	&\mathop {\min }\limits_{\bf{w}} \,\, {\left\| {\bf{w}} \right\|^2}
	\nonumber\\
	& \,\, \text{s.t.}   \,\,\,  {\mathop{\rm Im}\nolimits} \left( {{\bf{h}}_{{n_r}}^T{\bf{w}}{e^{i{\varphi _{{n_r}}}}}} \right) \ge  b_1 \, {\mathop{\rm Re}\nolimits} \left( {{\bf{h}}_{{n_r}}^T{\bf{w}}{e^{i{\varphi _{{n_r}}}}}} \right) + {a_1},
	\label{subeq:relax:1}
	\\
	&    \qquad   {\mathop{\rm Im}\nolimits} \left( {{\bf{h}}_{{n_r}}^T{\bf{w}}{e^{i{\varphi _{{n_r}}}}}} \right) \le b_2 \, {\mathop{\rm Re}\nolimits} \left( {{\bf{h}}_{{n_r}}^T{\bf{w}}{e^{i{\varphi _{{n_r}}}}}} \right) + {a_2},
	\label{subeq:relx:2}
	\end{align}
	\label{eqn:opt:relax:1}%
\end{subequations}
for $r = 1,...,R$ and $n=1,...,N$, where 
\begin{align}
{a_1}& = {c_1} - \sqrt {\left( {{{\cos }^{ - 2}}\left( {{\varphi _{{s_0}}} - \frac{\pi }{M}} \right) - 1} \right)c_2^2},
\nonumber\\
{a_2} &=  - \tan \left( {{\varphi _{{s_0}}} + \frac{\pi }{M}} \right)\left[ {{c_2} - \sqrt {\left( {{{\sin }^{ - 2}}\left( {{\varphi _{{s_0}}} + \frac{\pi }{M}} \right) - 1} \right)c_1^2} } \right],
\nonumber\\
{b_1} &= \tan \left( {{\varphi _{{s_0}}} - \frac{\pi }{M}} \right), \,\, {b_2} = \tan \left( {{\varphi _{{s_0}}} + \frac{\pi }{M}} \right),
\nonumber\\
{c_1} &= \sqrt \gamma  \sin \left( {\arg \left( {{s_0}} \right)} \right), \,\,    {c_2} = \sqrt \gamma  \cos \left( {\arg \left( {{s_0}} \right)} \right),
\label{eqn:def}
\end{align}
and ${\varphi _{{n_r}}} = \arg \left( {{s_{{n_r}}}} \right)$. The value of $\varphi _{{n_r}}$ can be absorbed in the channel to rewrite~\eqref{eqn:opt:relax:1} as
\begin{subequations}
	\begin{align}
	&\mathop {\min }\limits_{\bf{w}} \,\, {\left\| {\bf{w}} \right\|^2}
	\nonumber\\
	& \,\, \text{s.t.}   \,\,\,  {\rm{Im}}\left( {\widetilde {\bf{h}}_{{n_r}}^T{\bf{w}}} \right) \ge {b_1} \, {\rm{Re}}\left( {\widetilde {\bf{h}}_{{n_r}}^T{\bf{w}}} \right) + {a_1},
	\label{subeq:relax:11}
	\\
	&    \qquad  {\rm{Im}}\left( {\widetilde {\bf{h}}_{{n_r}}^T{\bf{w}}} \right) \le {b_2} \, {\rm{Re}}\left( {\widetilde {\bf{h}}_{{n_r}}^T{\bf{w}}} \right) + {a_2}.
	\label{subeq:relx:12}
	\end{align}
	\label{eqn:opt:relax:2}%
\end{subequations}
By stacking the constraints, we can encapsulate~\eqref{eqn:opt:relax:2} as 
\begin{subequations}
	\begin{align}
	&\mathop {\min }\limits_{\bf{w}} \,\, {\left\| {\bf{w}} \right\|^2}
	\nonumber\\
	& \,\, \text{s.t.}   \,\,\,    {\rm{Im}}\left( {\widetilde {\bf{H}}_{{U}}{\bf{w}}} \right) \ge {b_1} \, {\rm{Re}}\left( {\widetilde {\bf{H}}_{{U}}{\bf{w}}} \right) + {a_1} {\bf{1}},
	\label{subeq:relax:31}
	\\
	&    \qquad  {\rm{Im}}\left( {\widetilde {\bf{H}}_{{U}}{\bf{w}}} \right) \le {b_2} \, {\rm{Re}}\left( {\widetilde {\bf{H}}_{{U}}{\bf{w}}} \right) + {a_2} {\bf{1}},
	\label{subeq:relx:32}
	\end{align}
	\label{eqn:opt:relax:3}%
\end{subequations}
where $\bf{1}$ is an $N_U \times 1$ unit vector. We can use the relations developed in~\eqref{eqn:rel:ima:Hw} to transform~\eqref{eqn:opt:relax:3} into
\begin{align}
&\mathop {\min }\limits_{\bf{w}} \,\, {\left\| {\bf{w}} \right\|^2}
\nonumber\\
& \,\, \text{s.t.}   \,\,\,   {{\bf{B}}_1}{\bf{w}} \ge {\bf{a}},
\label{eqn:opt:relax:4}
\end{align}
where
\begin{align}
{{\bf{B}}_1} = \left[ \begin{array}{l}
\widetilde {\bf{H}}_{{U_2}} - {b_1}\widetilde {\bf{H}}_{{U_1}}\\
{b_2}\widetilde {\bf{H}}_{{U_1}} - \widetilde {\bf{H}}_{{U_2}}
\end{array} \right], \,\, {\bf{a}} = \left[ \begin{array}{l}
{a_1}{\bf{1}}\\
- {a_2}{\bf{1}}
\end{array} \right].
	\label{def:B1}
\end{align}
Using a similar approach as in Section~\ref{subsec:lar:tx},~\eqref{eqn:opt:relax:4} can be efficiently solved using the proposed iterative approach or the non-negative least squares formulation. 
}
\subsection{The Case of Strong Eavesdropper ($N_e \ge N_t$)}       
\label{subsec:lar:eve}
In this case, {\rc as the results in Section~\ref{sec:sim} show, $E$ can get a lower SER compared to the $N_e<N_t$ case.} This capability of $E$ comes from the fact that it has more antennas than $T$ and owns global CSI knowledge, which puts $E$ in a superior position compared to $T$ from hardware and CSI knowledge point of view. Nevertheless, there is still one possible way to enhance the security. {\rc Focusing on the signal part and ignoring the noise, we can see from~\eqref{eqn:eve:mim}, for ZF estimator, or~\eqref{eqn:mmse:est:sym}, for MMSE estimator, that ${\bf{\hat w}} = {\bf{w}}$. This means that} the estimated symbols by $E$ are equal to those induced on receiver antennas, {\rc ${{\bf{H}}_U}{\bf{w}}$, for the noiseless case,} therefore, we can design the precoder such that the SNR of the received signal becomes equal to the required level for successful decoding, which is defined by ACM. 

{\rc As the results of the case $N_e < N_t$ in Section~\ref{sec:sim} shows}, the SNR level at $E$ is lower than that of the users, which may prevent successful decoding of the $\mathit{M}$-PSK symbol at $E$. Based on this, we can minimize the sum power of the received signals at the users, ${\left\| {{{\bf{H}}_U}{\bf{w}}} \right\|^2}$, which is the same as the sum power of the estimated signals at $E$. {\rc In this frame, minimizing the sum power of the received signals is equivalent to minimizing the power of received signal on each receiving antenna.} Since the power of the received signal on each receiving antenna is constrained, {\rc minimizing the sum power results in the minimum possible power on each receiving antenna.} This results in a sort of \textit{``security fairness''} among the users. The precoder design problem for the signal level minimization precoder can be defined as
\begin{subequations}
\begin{align}
&\mathop {\min }\limits_{\bf{w}} \,\, {\left\| {{\bf{H}}_U} {\bf{w}} \right\|^2}
\nonumber\\
& \,\, \text{s.t.}   \,\,\,\,  \arg \left( {  {\bf{h}}_{n_r}^T  {\bf{w}}} \right) = \arg \left( {s_{n_r}} \right),                                                                \label{subeq:Hw:c1}
\\
&             \qquad \,\,     {\rm{Re}}\left( {s_{n_r}} \right){\rm{Re}}\left( {  {\bf{h}}_{n_r}^T  {\bf{w}}  } \right) \ge \sqrt \gamma {\rm{R}}{{\rm{e}}^2}\left( {s_{n_r}} \right),  \label{subeq:Hw:c2}
\end{align}
\label{eqn:opt:Hw:1}%
\end{subequations}
for $r = 1,...,R$ and $n=1,...,N$. {\rc Similar as in~\eqref{eqn:opt:pow:1}, the phase of the received signal on each receiving antenna in~\eqref{eqn:opt:Hw:1} is fixed, hence, we need to consider the signal level constraint on the real or imaginary part of the received signal.} Following a similar procedure as in Section~\ref{subsec:lar:tx},~\eqref{eqn:opt:Hw:1} can be transformed to
\begin{align}
&\mathop {\min }\limits_{\bf{w}} \,\, {\left\| {{\bf{H}}_U} {\bf{w}} \right\|^2}
\nonumber\\
& \,\, \text{s.t.}   \,\,\,\,\,     {\bf{A}}{\rm{Re}}\left( {{{\bf{H}}_U}{\bf{w}}} \right) - {\rm{Im}}\left( {{{\bf{H}}_U}{\bf{w}}} \right) = {\bf{0}},    
\nonumber\\
&                    \qquad \,\,  {\mathop{\rm Re}\nolimits} \left( {\bf{S}} \right){\mathop{\rm Re}\nolimits} \left( {{{\bf{H}}_U}{\bf{w}}} \right) \ge \sqrt {\gamma} \,  {{\bf{s}}_{r}},
\label{eqn:opt:Hw:2}
\end{align}
Using~\eqref{eqn:Hw} to~\eqref{eqn:rel:ima:Hw}, we expand ${\left\| {{{\bf{H}}_U}{\bf{w}}} \right\|^2}$ as
\begin{align}
{\left\| {{{\bf{H}}_U}{\bf{w}}} \right\|^2} &= {\widetilde {\bf{w}}^T}{\bf{H}}_{{U_1}}^T{{\bf{H}}_{{U_1}}}\widetilde {\bf{w}} + {\widetilde {\bf{w}}^T}{\bf{H}}_{{U_2}}^T{{\bf{H}}_{{U_2}}}\widetilde {\bf{w}}
\nonumber\\
&= {\widetilde {\bf{w}}^T}\left( {{\bf{H}}_{{U_1}}^T{{\bf{H}}_{{U_1}}} + {\bf{H}}_{{U_2}}^T{{\bf{H}}_{{U_2}}}} \right)\widetilde {\bf{w}},
\label{eqn:Hw:exp}
\end{align}
which along with~\eqref{eqn:rel:ima:Hw} helps us convert~\eqref{eqn:opt:Hw:2} into
\begin{align}
& \mathop {\min }\limits_{\widetilde {\bf{w}}}  \,\,\,\,  {\widetilde {\bf{w}}^T}\left( {{\bf{H}}_{{U_1}}^T{{\bf{H}}_{{U_1}}} + {\bf{H}}_{{U_2}}^T{{\bf{H}}_{{U_2}}}} \right)\widetilde {\bf{w}}
\nonumber\\
& \, \text{s.t.}   \,\,\,\,\,\,\,  \left( {{\bf{A}}{{\bf{H}}_{{U_1}}} - {{\bf{H}}_{{U_2}}}} \right)\widetilde {\bf{w}} =\bf{0},
\nonumber\\
&                  \qquad \,\,\,\,   {\mathop{\rm Re}\nolimits} \left( {\bf{S}} \right){{\bf{H}}_{{U_1}}}\widetilde {\bf{w}} \ge \sqrt {\gamma} \,  {\rc {{\bf{s}}_r}}.
\nonumber\\
\label{eqn:opt:Hw:3}
\end{align}
For~\eqref{eqn:opt:Hw:3} to be feasible, $\widetilde {\bf{w}}$ has to be in the null space of ${{\bf{A}}{{\bf{H}}_{{U_1}}} - {{\bf{H}}_{{U_2}}}}$. Hence, we can write $\widetilde {\bf{w}}$ as a linear combination of the null space basis of ${{\bf{A}}{{\bf{H}}_{{U_1}}} - {{\bf{H}}_{{U_2}}}}$ {\rc yielding} $\widetilde {\bf{w}} = {\bf{E}} \boldsymbol \lambda$, where ${\bf{E}}$ and $\boldsymbol{\lambda}$ are as in~\eqref{eqn:E:lam}. This way,~\eqref{eqn:opt:Hw:3} boils down to\footnote{The design problem~\eqref{eqn:opt:Hw:4} can be extended to M-QAM modulation by changing the constraint into equality. A detailed derivation falls beyond the scope of this paper.} 
\begin{align}
&\mathop {\min }\limits_{\boldsymbol{\lambda}} \,\, {\boldsymbol{\lambda} ^T}{{\bf{E}}^T}\left( {{\bf{H}}_{{U_1}}^T{{\bf{H}}_{{U_1}}} + {\bf{H}}_{{U_2}}^T{{\bf{H}}_{{U_2}}}} \right){\bf{E}} \boldsymbol{\lambda}
\nonumber\\
& \, \text{s.t.}   \,\,\,\,\,\,  {{\bf{B} \boldsymbol{\lambda} }} \ge \sqrt \gamma {\mkern 1mu} {\rc {{\bf{s}}_r}},
\label{eqn:opt:Hw:4}
\end{align}
{\rc where ${\bf{B}} = {\rm{Re}}\left( {\bf{S}} \right){{\bf{H}}_{{U_1}}}{\bf{E}}$}. Similar as in Section~\ref{subsec:lar:tx}, in the following, we propose an iterative algorithm and non-negative least squares formulation to solve~\eqref{eqn:opt:Hw:4}.
\vspace{3pt}
\subsubsection{Iterative solution} 
\label{subsubsec:lar:eve:sol:a}
By introducing the new variable $\bf{u}$, we can rewrite~\eqref{eqn:opt:Hw:4} as  
\begin{align}
& \mathop {\min }\limits_{{\boldsymbol{\lambda }},{\bf{u}}}  \,\, {\boldsymbol{\lambda} ^T}{{\bf{E}}^T}\left( {{\bf{H}}_{{U_1}}^T{{\bf{H}}_{{U_1}}} + {\bf{H}}_{{U_2}}^T{{\bf{H}}_{{U_2}}}} \right){\bf{E}} \boldsymbol{\lambda}
\nonumber\\
& \, \text{s.t.}   \,\,\,\,\,\,  {{\bf{B} \boldsymbol{\lambda} }} = \sqrt \gamma  {\rc {{\bf{s}}_r}} + \bf{u}.
\label{eqn:opt:Hw:5}
\end{align}
We can adapt Algorithm~\ref{alg:itr} to solve~\eqref{eqn:opt:Hw:4} by replacing the solution to ${\boldsymbol \lambda ^\star}$ as 
\begin{align}
{\boldsymbol{\lambda} ^\star}  = {\left( {\frac{{{{\bf{E}}^T}\left( {{\bf{H}}_{{U_1}}^T{{\bf{H}}_{{U_1}}} + {\bf{H}}_{{U_2}}^T{{\bf{H}}_{{U_2}}}} \right){\bf{E}}}}{\eta } + {{\bf{B}}^T}{\bf{B}}} \right)^{ - 1}}{{\bf{B}}^T}\left( {{\bf{a}} + {\bf{u}}} \right),
\label{eqn:opt:Hw:lam:star}
\end{align}
which is derived using a similar procedure as in Section~\ref{subsubsec:lar:tx:sol:a}. Similar as in~\eqref{eqn:opt:itr:lam:cf}, the matrix inversion in~\eqref{eqn:opt:Hw:lam:star} needs to be calculated only once per symbol transmission.
\vspace{3pt}
\subsubsection{Non-negative least squares} 
\label{subsubsec:lar:eve:sol:b}
Assuming that ${{\bf{H}}_{{U_1}}}$ and ${{\bf{H}}_{{U_2}}}$ are non-singular, the matrix ${{\bf{E}}^T}\left( {{\bf{H}}_{{U_1}}^T{{\bf{H}}_{{U_1}}} + {\bf{H}}_{{U_2}}^T{{\bf{H}}_{{U_2}}}} \right){\bf{E}}$ is positive definite, hence, its Cholesky decomposition ${{\bf{E}}^T}\left( {{\bf{H}}_{{U_1}}^T{{\bf{H}}_{{U_1}}} + {\bf{H}}_{{U_2}}^T{{\bf{H}}_{{U_2}}}} \right){\bf{E}} = {\bf{L}}{{\bf{L}}^T}$ exists and can be used in order to rewrite~\eqref{eqn:opt:Hw:5} as
\begin{align}
& \mathop {\min }\limits_{{\boldsymbol{\lambda }},{\bf{u}}} \,\, {\left\| {{{\bf{L}}^T}{\bf{\boldsymbol\lambda }}} \right\|^2}
\nonumber\\
& \, \text{s.t.}   \,\,\,\,\,\,  {{\bf{B} \boldsymbol{\lambda} }} = \sqrt \gamma  {\rc {{\bf{s}}_r}} + \bf{u}.
\label{eqn:opt:Hw:6}
\end{align}
We can derive $\boldsymbol{\lambda}$ using the constraint of~\eqref{eqn:opt:Hw:6} as ${\boldsymbol{\lambda }} = {{\bf{B}}^\dag }\left( {\sqrt \gamma  {\rc {{\bf{s}}_r}} + {\bf{u}}} \right)$ and replace it back into the objective of~\eqref{eqn:opt:Hw:6} to get
\begin{align}
& \mathop {\min }\limits_{\bf{u}} \,\, {\left\| {{{\bf{L}}^T}{{\bf{B}}^\dag }{\bf{u}} + {{\bf{L}}^T}{{\bf{B}}^\dag }\sqrt \gamma  {\rc {{\bf{s}}_r}}} \right\|^2}
\nonumber\\
& \, \text{s.t.}   \,\,\,\,\,\,   \bf{u} \ge 0,
\label{eqn:opt:Hw:7}
\end{align}
which is a non-negative least squares optimization problem. Since ${{{\bf{L}}^T}{{\bf{B}}^\dag }}$ and ${{{\bf{L}}^T}{{\bf{B}}^\dag }\sqrt \gamma  {\rc {{\bf{s}}_r}}}$ are real valued, we can use~\cite{Solving:1995,Bro1997b} to solve~\eqref{eqn:opt:Hw:7} in an efficient way.
\subsection{{\rc Benchmark Scheme}} 
\label{sub:bm}
We consider the ZF at the transmitter~\cite{Lai-U:2004} as the benchmark scheme since both our design and the benchmark scheme use the CSI knowledge at the transmitter to design the precoder.

In the benchmark scheme, ZF precoder is applied at the transmitter to remove the interference among the symbol streams. The received signals at users and $E$ in the benchmark scheme are
\begin{align}
&{{\bf{y}}_U^\Q} = {{\bf{H}}_U}{\bf{W}}^{} \bf{s} \beta + {{\bf{n}}_U^\Q}, 
\\
&{{\bf{y}}_E^\Q} = {{\bf{H}}_E}{\bf{W}}^{} \bf{s} \beta + {{\bf{n}}_E^\Q},
\label{eqn:rec:sig:zf}
\end{align}
where ${\bf{W}}^{} = {\bf{H}}_U^H{\left( {{\bf{H}}_U{\bf{H}}_U^H} \right)^{ - 1}}$ is the precoding vector, $\bf{s}$ contains the symbols, and $\beta$ is the amplification factor for the symbols which acts similar as $\sqrt \gamma$ in the directional modulation scheme. For a fair comparison, we pick up the same values for $\sqrt \gamma$ and $\beta$ in the simulations. 

When using the benchmark, $E$ can {\rc use ZF and MMSE as two possible ways to estimate the symbols. In contrast to our method} $E$ can use the knowledge of ${\bf{H}}_U$ to calculate $\bf{W}$ {\rc in the benchmark scheme.}

In the {\rc ZF approach}, given that $N_e \ge N_t$, $E$ can estimate ${\bf{s}}\beta$ as
\begin{align}
{\widehat {{\bf{s}}\beta }} &= {\left[ {{{\left( {{{\bf{H}}_E}{\bf{W^{}}}} \right)}^H}{{\bf{H}}_E}{\bf{W^{}}}} \right]^{ - 1}}{\left( {{{\bf{H}}_E}{\bf{W^{}}}} \right)^H}{{\bf{y}}_E^\Q}
\nonumber\\
&= {\bf{s}}\beta  + {\left[ {{{\left( {{{\bf{H}}_E}{\bf{W^{}}}} \right)}^H}{{\bf{H}}_E}{\bf{W^{}}}} \right]^{ - 1}}{\left( {{{\bf{H}}_E}{\bf{W^{}}}} \right)^H}{{\bf{n}}_E^\Q}
\label{eqn:eve:2nd}
\end{align}
where ${\widehat {{\bf{s}}\beta }}$ is the estimated ${\bf{s}}\beta$ at $E$. Since ${{\bf{H}}_E}{\bf{W}}$ is $N_e \times N_{U}$, ${\left[ {{{\left( {{{\bf{H}}_E}{\bf{W}}} \right)}^H}{{\bf{H}}_E}{\bf{W}}} \right]^{ - 1}}{\left( {{{\bf{H}}_E}{\bf{W}}} \right)^H}{{\bf{H}}_E}{\bf{W}}=\bf{I}$ for $N_e \ge N_{U}$. Hence, in the benchmark scheme, $E$ can derive the precoder and estimate the symbols {\rc using the ZF method} when $N_e \ge N_{U}$. On the other hand, since our designed precoder depends on both the channels and symbols, $E$ cannot derive the precoder and estimate the symbols {\rc using the ZF method} when $N_e \ge N_{U}$. 

{\rc  In the MMSE approach, $E$ can estimate ${\bf{s}}\beta$ as
\begin{align}
& {\widehat {{\bf{s}}\beta }} = {{\bf{G}}_3}{{\bf{y}}_E^\Q},
\label{eqn:mmse:est:bm}
\end{align}
where
\begin{align}
{{\bf{G}}_3} = {\left[ {{{\left( {{{\bf{H}}_E}{\bf{W}}} \right)}^H}{\bf{C}}_{\bf{w}}^{ - 1}{{\bf{H}}_E}{\bf{W}} + {\bf{C}}_{{{\bf{N}}_E}}^{ - 1}} \right]^{ - 1}}{\left( {{{\bf{H}}_E}{\bf{W}}} \right)^H}{\bf{C}}_{\bf{w}}^{ - 1}.
\label{eqn:mmse:mat:bm}
\end{align}}%
{\rc When using the benchmark method, we will see in Section~\ref{sec:sim} that SER at $E$ when using the MMSE method depends on the difference between $N_e$ and $N_U$, while the SER at $E$ depends on the difference between $N_e$ and $N_t$ in our method.} Broadly speaking, the base station has usually more antennas than the users, hence, {\rc it is more likely to have a higher difference between $N_e$ and $N_t$ rather than $N_e$ and $N_U$}, especially with a large-scale array. Therefore, it is more probable to preserve the security in our design compared to the benchmark scheme. Furthermore, by comparing~\eqref{eqn:eve:mim} {\rc and~\eqref{eqn:mmse:est:sym} with}~\eqref{eqn:eve:2nd}, we see that $E$ has to multiply ${\widehat {\bf{W}}}$ by ${{\bf{H}}_U}$ in our design whereas $E$ does need to do this in the benchmark scheme.  
{\rc \section{Remarks on Computational Complexity} 
\label{sec:complex}
In this part, we analyze the computational complexity of our method and the benchmark scheme assuming that we pick up the non-negative formulation approach to design our precoder. The computational complexity of the non-negative least squares approach when using the interior point,~\eqref{eqn:com:comp:ip}, and fast projected gradient algorithms,~\eqref{eqn:com:comp:fpg}, are, respectively, as~\cite{FGP:2015}
\begin{align}
&O\left( {N_t^3\ln {\varepsilon ^{ - 1}}} \right),
\label{eqn:com:comp:ip}
\\
&O\left( {{\lambda_0 ^{\frac{1}{2}}}\left\| {{{\bf{w}}_0} - {{\bf{w}}^\star}} \right\|N_t^2{\varepsilon ^{ - \frac{1}{2}}}} \right),
\label{eqn:com:comp:fpg}
\end{align}
where $\varepsilon$ is the upper bound on the difference between the current, $f\left( {{{\bf{w}}_{itr}}} \right)$, and the optimal value, $f\left( {{{\bf{w}}^\star}} \right)$, of the objective function as $f\left( {{{\bf{w}}_{itr}}} \right) - f\left( {{{\bf{w}}^\star}} \right) \le \varepsilon$, ${\lambda _0} = {\lambda _{\max }}\left( {{{\bf{D}}^T}{\bf{D}}} \right)$ with ${\bf{D}} = {{\bf{B}}^\dag }$ for~\eqref{eqn:opt:nnls}, ${\bf{D}} = {{\bf{B}}_1^\dag }$ for~\eqref{eqn:opt:relax:4}, and ${\bf{D}} = {{\bf{L}}^T}{{\bf{B}}^\dag }$ for~\eqref{eqn:opt:Hw:7}.

Next, we derive the computational complexity of the benchmark scheme. Considering the structure of $\bf{W}$, the complexity of the benchmark scheme is derived as
\begin{align}
2O\left( {{N_t}N_{_U}^2} \right) + O\left( {N_U^3} \right) + O\left( {{N_t}{N_U}} \right).
\label{eqn:bm:com}
\end{align}

Each of the problems in~\eqref{eqn:opt:nnls},~\eqref{eqn:opt:relax:4}, and~\eqref{eqn:opt:Hw:7}, need to be solved once per group of symbols communications. In other words, $N_U$ symbols can be communicated for each designed precoder. Therefore, a higher $N_U$ means that more symbols can be communicated to the users for each designed precoder. On the other hand, the designed precoder in the benchmark scheme can be used as far as the channel is fixed. Hence, the computational complexity comparison between our scheme and the benchmark method depends on the channel changing rate, the total number of users' antennas, and the required accuracy in the non-negative least squares solution in~\eqref{eqn:com:comp:ip} or~\eqref{eqn:com:comp:fpg}.   
}
\section{Simulation Results}   
\label{sec:sim}
In this part, we present different simulation scenarios to analyze the security and the performance of the directional modulation scheme for different precoding designs, and compare them with a benchmark scheme. In all simulations, channels are considered to be quasi static block Rayleigh which are generated using i.i.d. complex Gaussian random variables with distribution $\mathcal{CN}(0 , 1 )$ and remain fixed during the interval that the $\mathit{M}$-PSK symbols are being induced at the receiver. Also, the noise is generated using i.i.d. complex Gaussian random variables with distribution $\mathcal{CN}(0 , {\sigma ^2} )$, and the modulation order used in all of the scenarios is $8$-PSK modulation. Here, we simulate each precoder for both strong transmitter, $N_e <N_t$, and strong eavesdropper, $N_e \ge N_t$, cases. This way, we show the benefit of the power minimizer precoder in the strong transmitter case and the signal level minimizer precoder in the strong eavesdropper case. We use the acronym ``min'' instead of minimization in the legend of the figures. {\rc Unless otherwise mentioned, the power minimization precoder used in the scenario is the one with fixed phase. Here, the SER at $E$ is derived by assuming that $E$ decodes the symbols of all users.}

In all the experiments, the computation times of the iterative method and non-negative least squares were considerably lower than the computation time of CVX. For example, in the case $N_t=20$ and $N_U=20$, while the average required time for the iterative method and non-negative least squares was $173.4$ and $10.5$ milliseconds, respectively, the same task was accomplished by CVX in $999.3$ milliseconds. 

In the first scenario, the effect of {\rc the} number of transmitter antennas, $N_t$, on transmitter's consumed power and the SER at users and $E$ are investigated for {\rc power minimization, fixed and relaxed phase, and signal level minimization precoders} in~\eqref{eqn:opt:pow:1},{\rc ~\eqref{eqn:opt:relax:4}}, and~\eqref{eqn:opt:Hw:1}, and the benchmark scheme. The average consumed power, ${\left\| {\bf{w}} \right\|^2}$, with respect to $N_t$ is shown in Fig.~\ref{fig:pow:ant} for $N_{U}=8, 10$. As $N_t$ increases, the power consumption of our design with power minimization {\rc precoders, fixed and relaxed phase, converge} to that of other two schemes. The power consumed by power minimization {\rc precoders with fixed and relaxed phase have the largest difference with the other two schemes, almost $6$ and $8$ dB}, for $N_t=N_{U}$. {\rc We see that power minimization precoder with relaxed phase has $2.5$ dB less power consumption compared to the power minimization precoder with fixed phase}. The signal level minimization precoder has almost the same power consumption as the benchmark scheme {\rc for $N_t=N_U=10$}. When the difference between $N_t$ and $N_{U}$ increases, all {\rc four} schemes consume considerably less power. {\rc When $N_t$ is larger than $N_U$, the degrees of freedom of the signal level minimization design increases and the power consumed by the signal level minimization precoder approaches that of the power minimization precoder.}

The average total SER at users and the average SER at $E$ with respect to $N_t$ are presented in {\rc Figures~\ref{fig:ser:ant1} and~\ref{fig:ser:ant2} where the eavesdropper uses ZF and MMSE to estimate the symbols}. Our designed precoders, power and signal level minimization, cause considerably more SER at $E$ compared to the benchmark scheme for a long range of $N_t$. Furthermore, as $N_e$ increases, there are cases, e.g., $N_t=16$, that the error caused at $E$ by the benchmark scheme decreases while the error caused by our designed precoders remains almost fixed {\rc when $E$ used the ZF estimator and reduces slightly when $E$ uses the MMSE estimator}. As Fig.~\ref{fig:rec:pow::ant} shows, our design with signal level minimization precoder and the benchmark scheme keep users' signal level norm constant. This leads into a constant SNR at $E$. 

{\rc We see in Figures~\ref{fig:ser:ant1} and~\ref{fig:ser:ant2} that the MMSE estimator results in a less SER at the eavesdropper compared to the ZF estimator when the difference between $N_t$ and $N_U$ increases. On the other hand, for close values of $N_t$ and $N_U$, the MMSE approach leads into the same SER as the ZF approach. Although the MMSE estimator reduces the SER at the eavesdropper, the error at the eavesdropper is still much higher than the users. For example, in Fig.~\ref{fig:ser:ant1}, the SER at the eavesdropper is $0.2$ while the SER at the users is $10^{-3}$. We see in Fig.~\ref{fig:ser:ant2} that for $N_t=N_U=10$, the eavesdropper can reduce the SER more in the benchmark scheme compared to our method.} 
Since the directional modulation with signal level minimization imposes more error on $E$ and consumes the same power as the benchmark scheme, it is the preferable choice for secure communication when $N_e \ge N_t$. Comparing Fig.~\ref{fig:pow:ant} {\rc with Figures~\ref{fig:ser:ant1} and~\ref{fig:ser:ant2}} shows that when the difference between $N_t$ and $N_{U}$ goes above a specific amount, the power and signal level minimization precoders converge in both power consumption and the SER at $E$ and users. 
\begin{figure}[t]
	\centering
	\includegraphics[width=9cm]{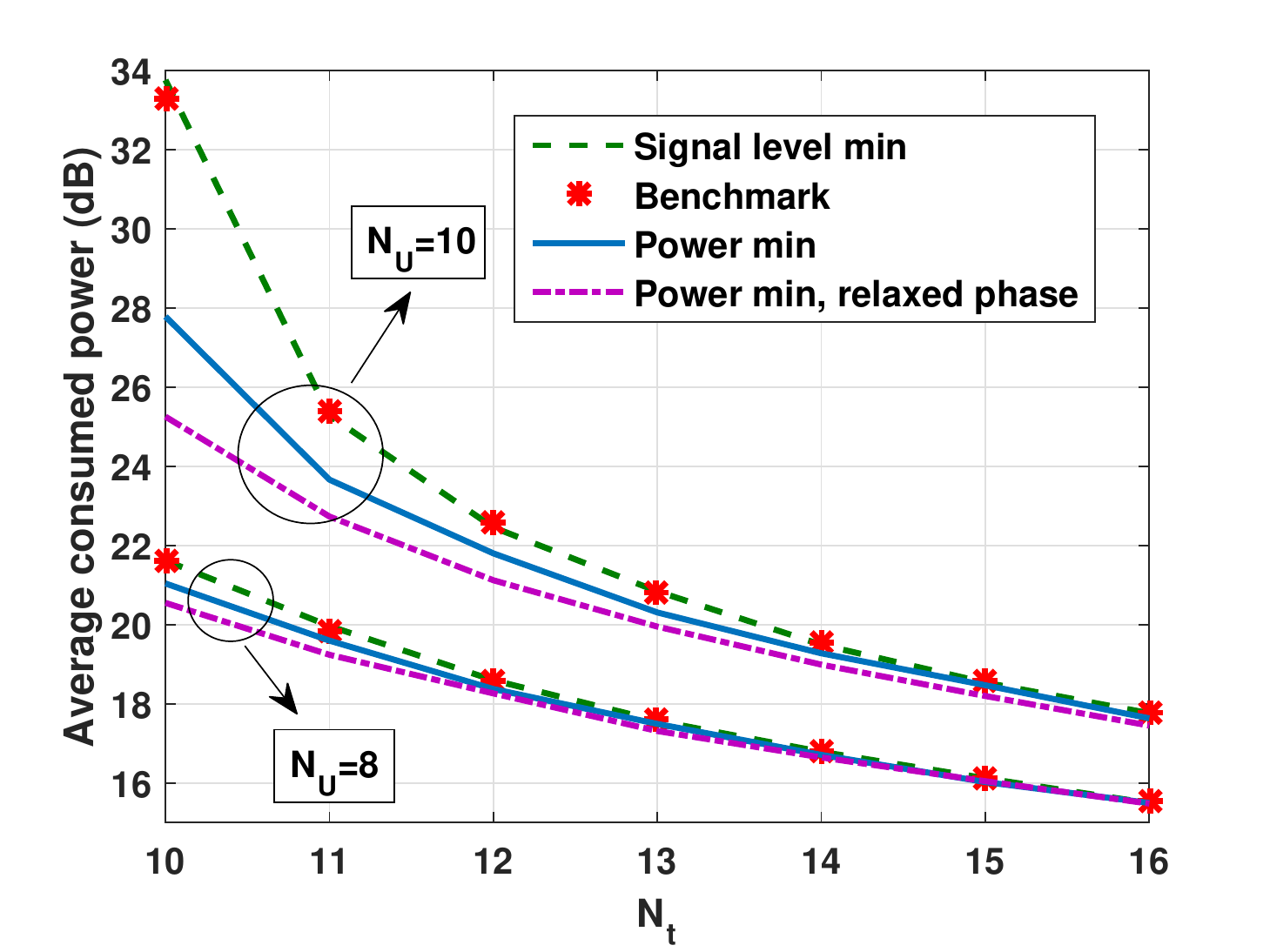}
	\caption{Average consumed power with respect to $N_t$ for our designed precoders and the benchmark scheme when $\gamma=15.56$ dB and $\beta^2=15.56$ dB.}  
	\label{fig:pow:ant}
\end{figure}
\begin{figure}[t]
	\centering
	\includegraphics[width=9cm]{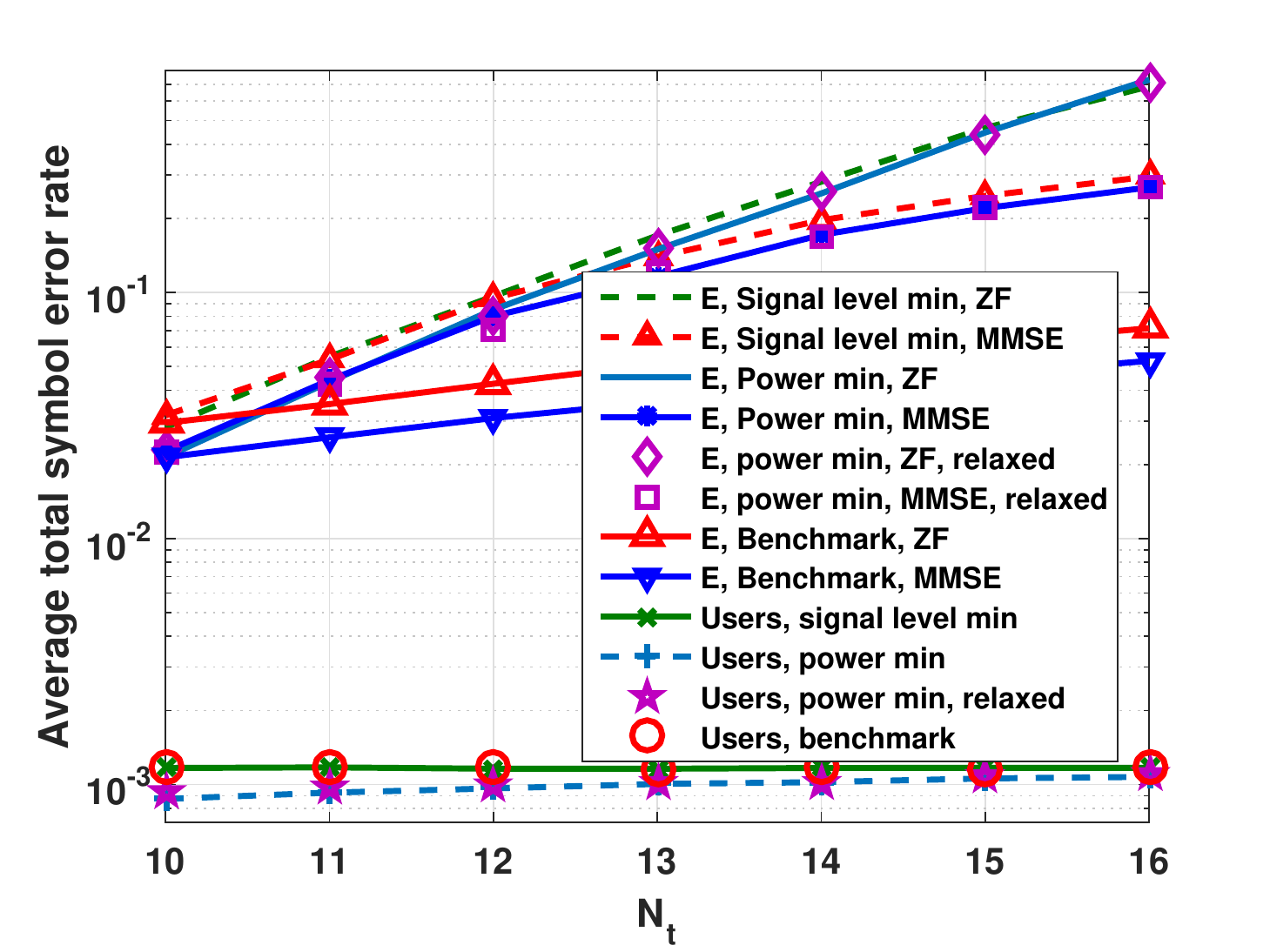}
	\caption{Average total SER at the users and average SER at $E$ with respect to $N_t$ for our designed precoders and the benchmark scheme when $N_{U}=10$, {\rc $N_e=15$}, $\gamma=15.56$ dB, and $\beta^2=15.56$ dB.}
	\label{fig:ser:ant1}
\end{figure}
\begin{figure}[t]
	\centering
	\includegraphics[width=9cm]{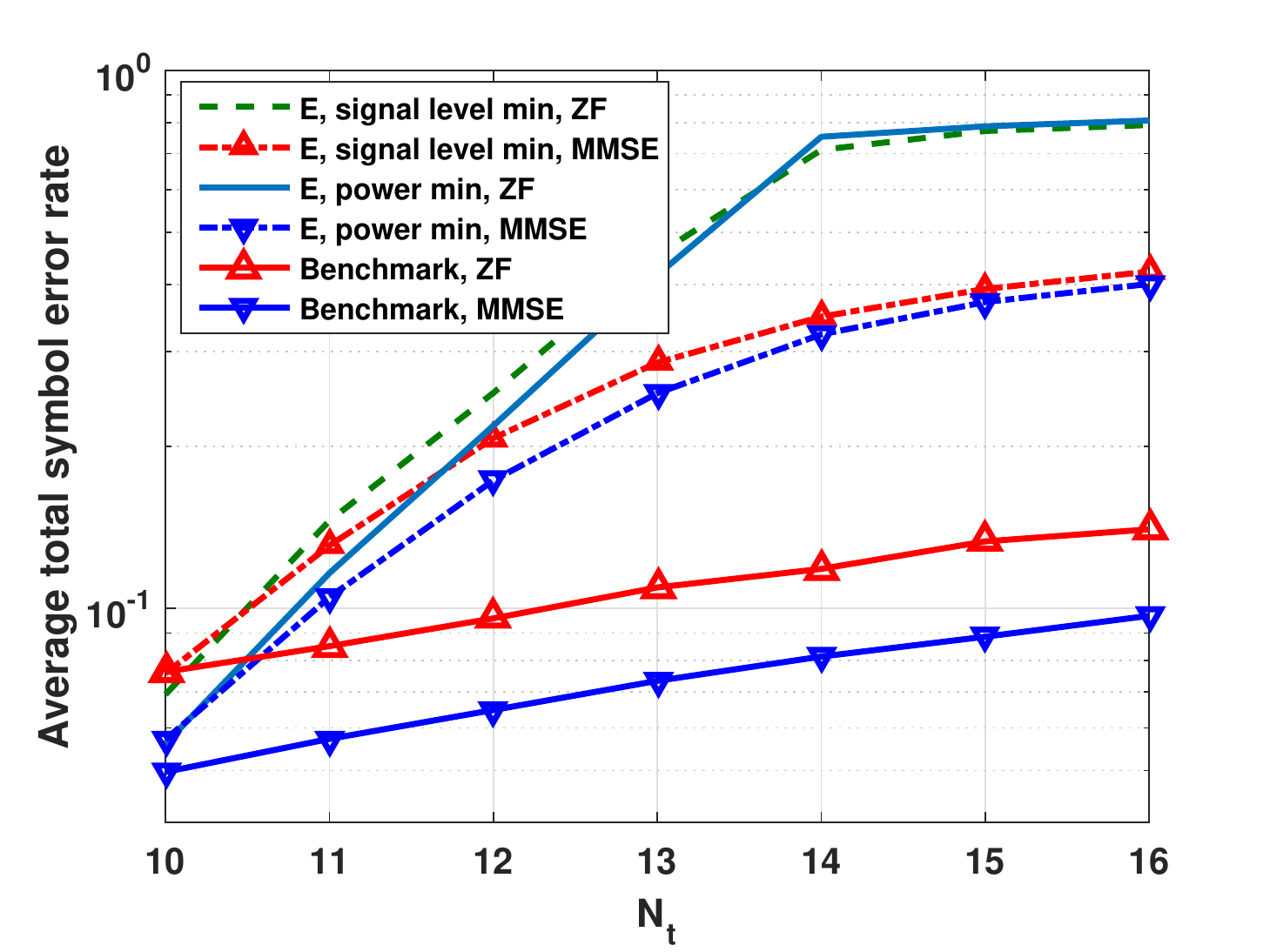}
	\caption{{\rc Average SER at $E$ with respect to $N_t$ for our designed precoders and the benchmark scheme when $N_{U}=10$, $N_e=13$, $\gamma=15.56$ dB, and $\beta^2=15.56$ dB.}}
	\label{fig:ser:ant2}
\end{figure}
\begin{figure}[t]
	\centering
	\includegraphics[width=9cm]{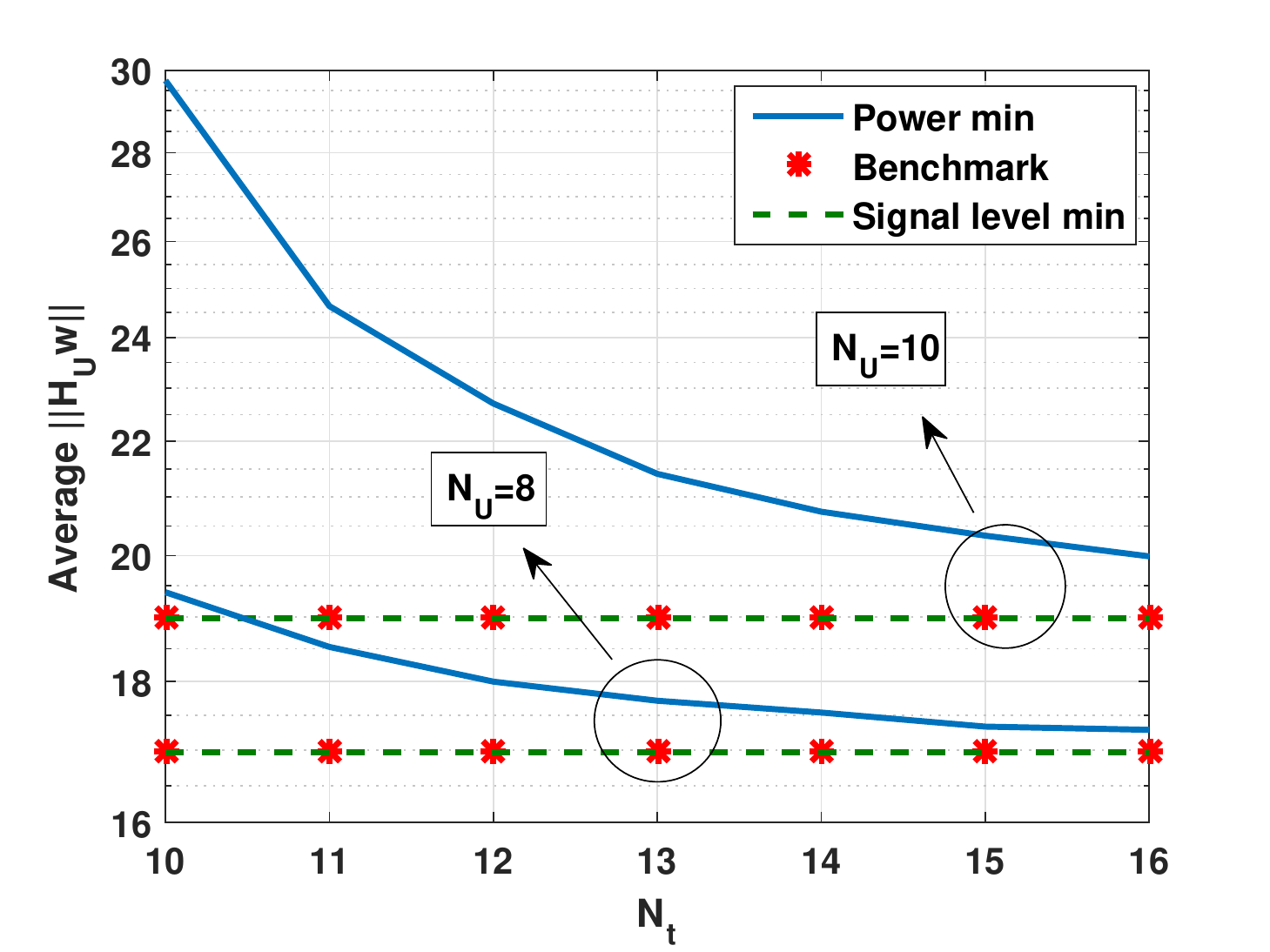}
	\caption{Average $\left\| {{{\bf{H}}_U}{\bf{w}}} \right\|$ for our designed precoders and the benchmark scheme when $\gamma=15.56$ dB, and $\beta^2=15.56$ dB.} 
	\label{fig:rec:pow::ant}
\end{figure}
\begin{figure}[]
	\centering
	\includegraphics[width=9cm]{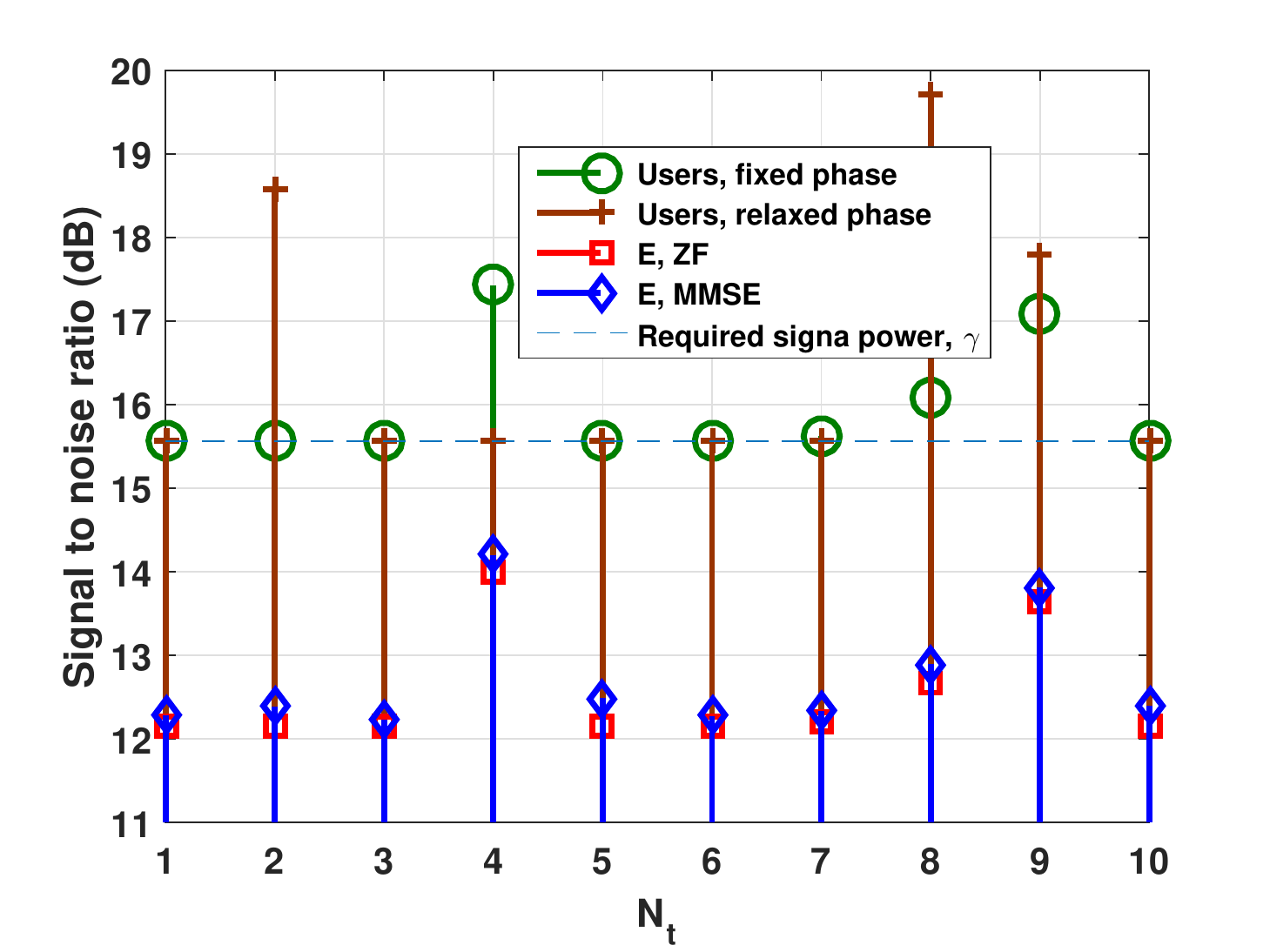}
	\caption{{\rc Instantaneous symbol power to average noise power for power minimization precoder with fixed and relaxed phase designs when $N_t=11$, $N_{U}=10$, $N_e=16$ and $ \gamma=15.56$ dB.}}
	\label{fig:snr:ins:1}
\end{figure}
\begin{figure}[]
	\centering
	\includegraphics[width=9cm]{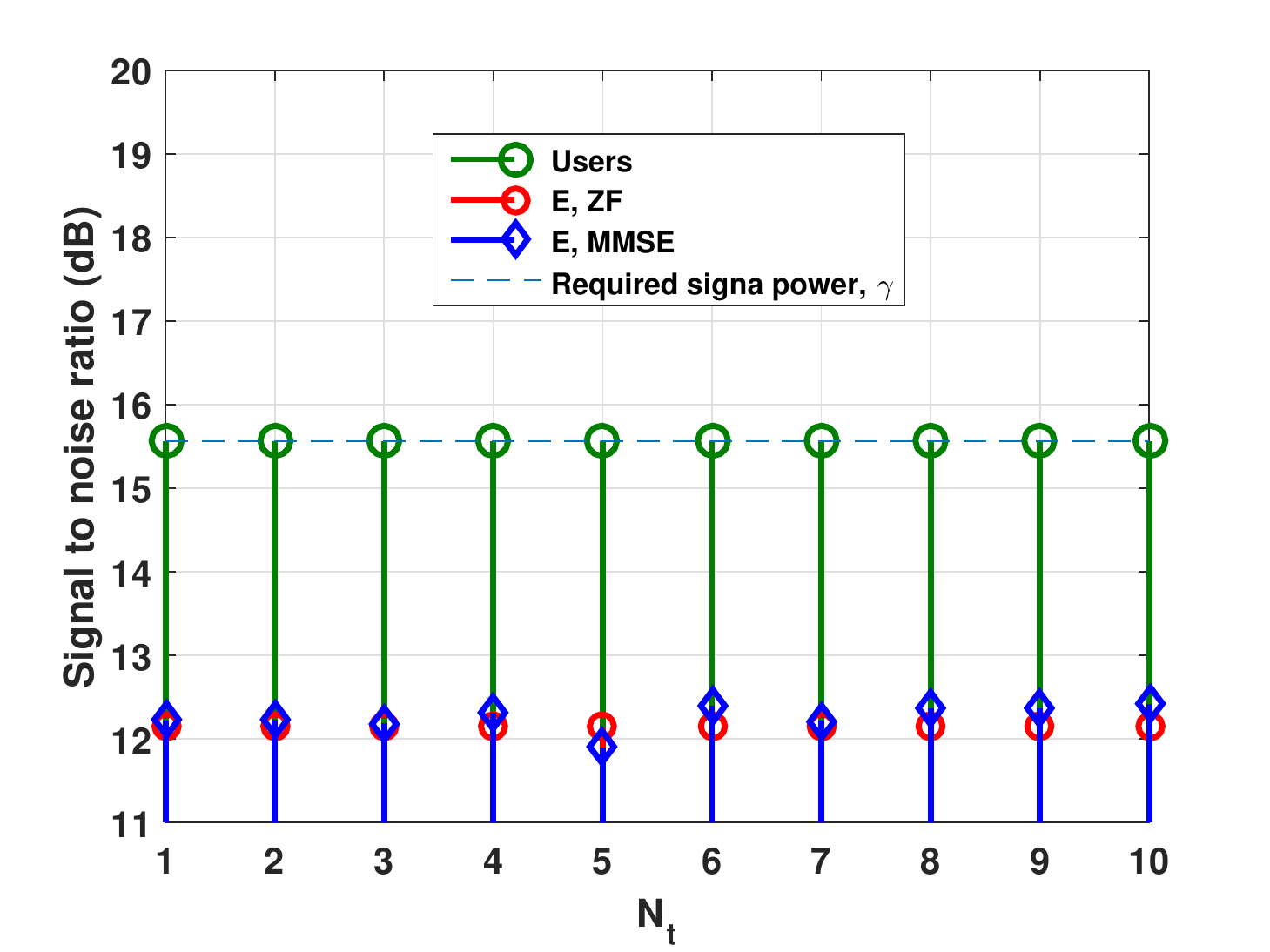}
	\caption{\rc Instantaneous symbol power to average noise power for signal level minimization precoder when $N_t=11$, $N_{U}=10$, $N_e=16$ and $\gamma=15.56$ dB.}
	\label{fig:snr:ins:2}
\end{figure}

The instantaneous power of the induced symbols to average noise power is shown in {\rc Figures~\ref{fig:snr:ins:1} and~\ref{fig:snr:ins:2} for power, fixed and relaxed phase,} and signal level minimization precoders when $N_e > N_t$. As we see, even with $E$ being able to estimate the symbols, the SNR at $E$ is lower than the users. {\rc This shows that the processes carried out at $E$ to perform ZF and MMSE estimations of $\bf{w}$ cause the SNR to be less than that of the users.} {\rc As Fig.~\ref{fig:snr:ins:2} shows, the signal level minimization precoder keeps the SNR at the users and $E$ at the lowest possible level. The SNR at the users is on the required threshold for decoding while the SNR at $E$ is much lower than that of the users and below the required threshold for successful decoding}, which imposes the maximum SER on $E$.

In the second scenario, $T$'s average power consumption, total average SER at the users, and average SER at $E$ are plotted with respect to total receiving antennas, $N_{U}$. Fig.~\ref{fig:pow:user:ant} shows the average consumed power with respect to $N_{U}$. Increasing $N_{U}$ decreases the degrees of freedom and increases the power consumption. As $N_{U}$ approaches $N_t$, the difference between the power consumed by the power minimization precoder and the other two schemes increases. 

We investigate the effect of $N_{U}$ on average total SER at the users and the average SER at $E$ {\rc in Figures~\ref{fig:ser:ant:user1} and~\ref{fig:ser:ant:user2}}. As $N_{U}$ increases, the SNR provided by the power minimization precoder goes more above the threshold. This reduces the average SER at both users and $E$. On the other hand, regardless of difference between $N_t$ and $N_{U}$, our design with signal level minimization precoder always preserves the SER at $E$ in the maximum value. {\rc Compared to the ZF estimator, when our precoders are used, the MMSE approach reduces the SER at $E$ for close values of $N_U$ and $N_t$. As $N_U$ approaches $N_t$, the performance of ZF and MMSE techniques get closer. As shown in Fig.~\ref{fig:ser:ant:user2}, the MMSE estimator at $E$ reduces the SER more compared to the ZF estimator when the signal level min precoder is used.} When $N_e > N_{U}$, our design imposes more SER at $E$ compared to the benchmark scheme since $N_e \ge N_{U}$ is required for $E$ to estimate the symbols in the benchmark scheme. As $N_{U}$ approaches $N_t$, the SER imposed on $E$ by the signal level minimization precoder and the benchmark scheme get closer.    
\begin{figure}[]
	\centering
	\includegraphics[width=9cm]{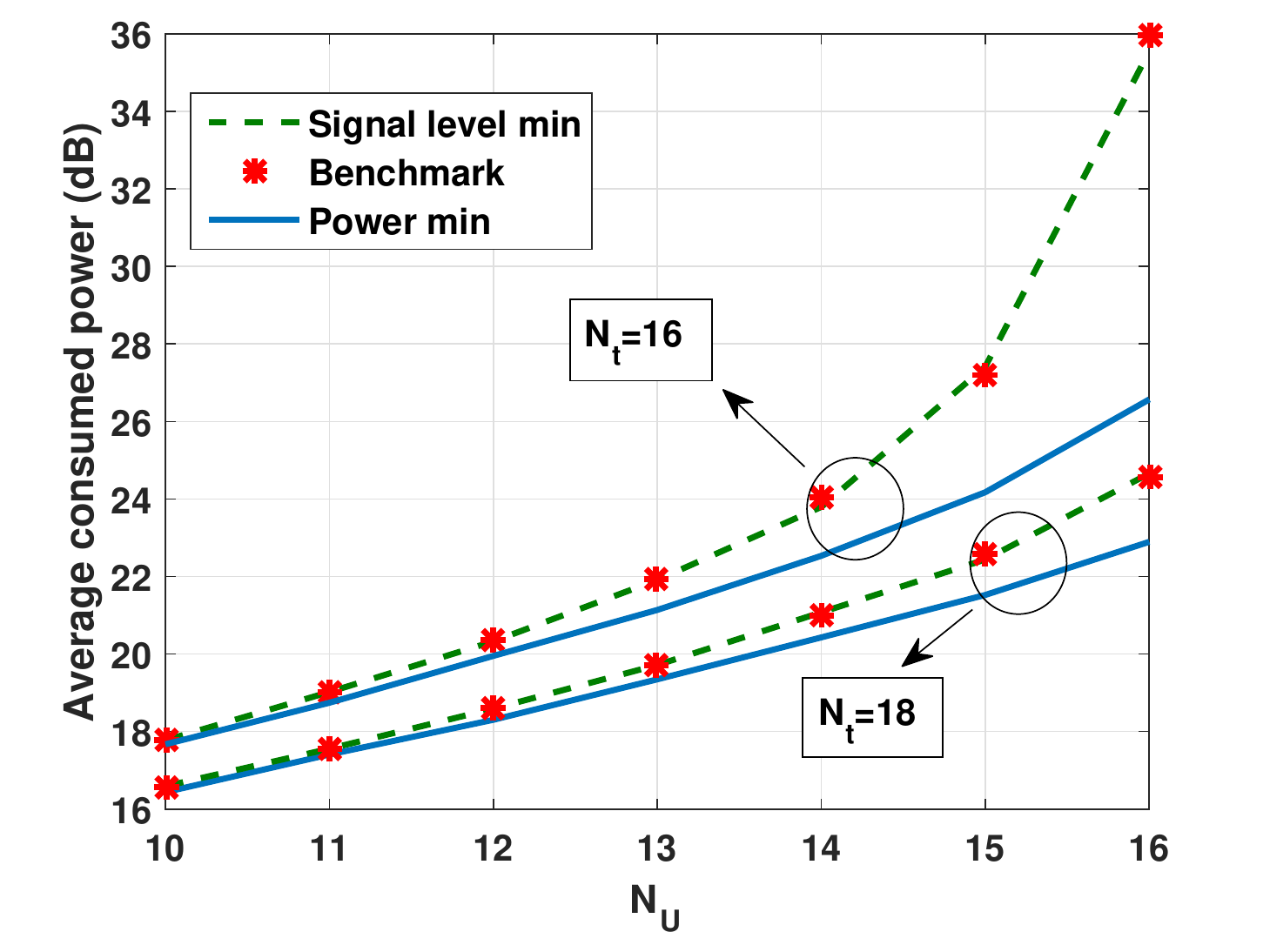}
	\caption{Average consumed power with respect to $N_{U}$ for our designed precoders and the benchmark scheme when $\gamma=15.56$~dB, and $\beta^2=15.56$~dB.}
	\label{fig:pow:user:ant}
\end{figure}
\begin{figure}[]
	\centering
	\includegraphics[width=9cm]{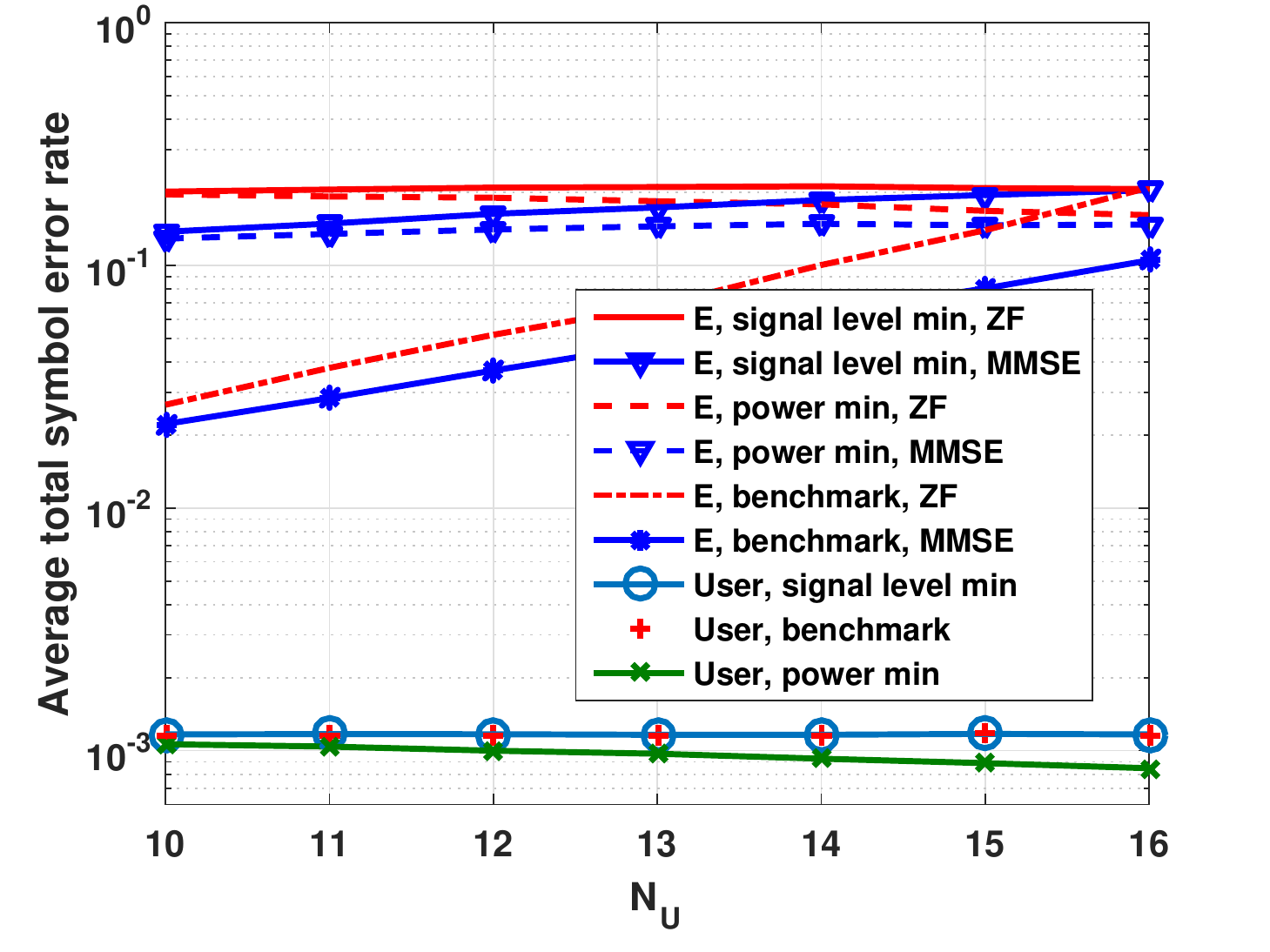}
	\caption{{\rc Average SER versus $N_{U}$ for our designed precoders and the benchmark scheme when $N_t=16$, $N_e=18$, $\gamma=15.56$~dB, and $\beta^2=15.56$~dB.}}
	\label{fig:ser:ant:user1}
\end{figure}
\begin{figure}[]
	\centering
	\includegraphics[width=9cm]{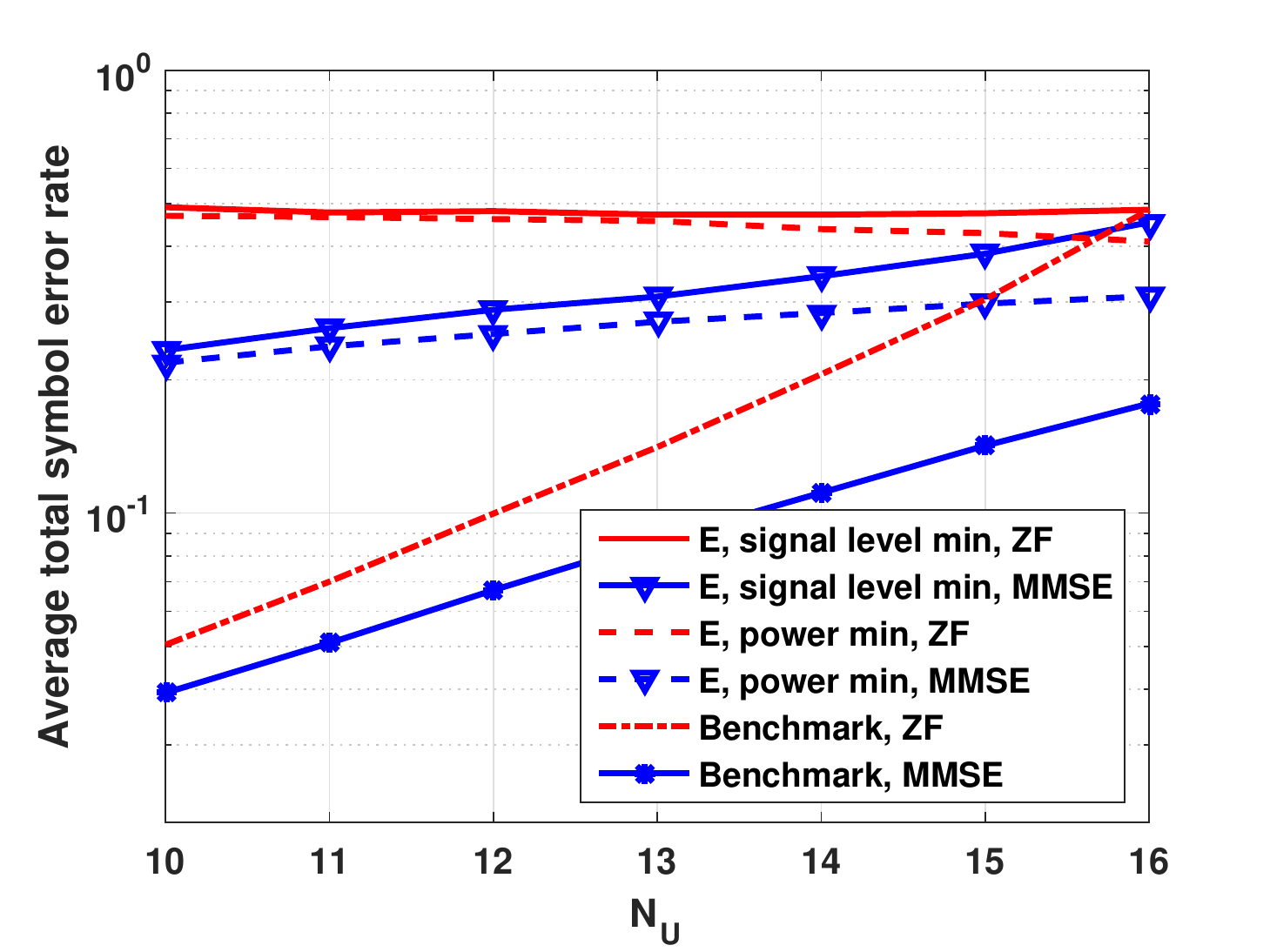}
	\caption{Average SER versus $N_{U}$ for our designed precoders and the benchmark scheme when $N_t=16$, {\rc $N_e=16$,} $\gamma=15.56$~dB, and $\beta^2=15.56$~dB.}
	\label{fig:ser:ant:user2}
\end{figure}

The next scenario inspects the effect of the required SNR for the received signals, $\gamma$, on $T$'s consumed power and the SER at users and $E$. Fig.~\ref{fig:pow:snr} shows the average consumed power with respect to $\gamma$ for our design and the benchmark scheme. The difference between the power consumed by the power minimization precoder and the other two schemes in low SNRs is more than that of high SNRs. The average total SER at users and the average SER at $E$ with respect to $\gamma$ is shown in Fig.~\ref{fig:ser:snr}. As SNR increases, {\rc the SER imposed on $E$ by our design becomes more than that of the benchmark scheme}. {\rc Furthermore, the performance of ZF and MMSE get closer as the SNR increases.} The difference between the average total SER at the users for power and signal level minimization precoders remains almost constant as $\gamma$ increases.
\begin{figure}[]
	\centering
	\includegraphics[width=9cm]{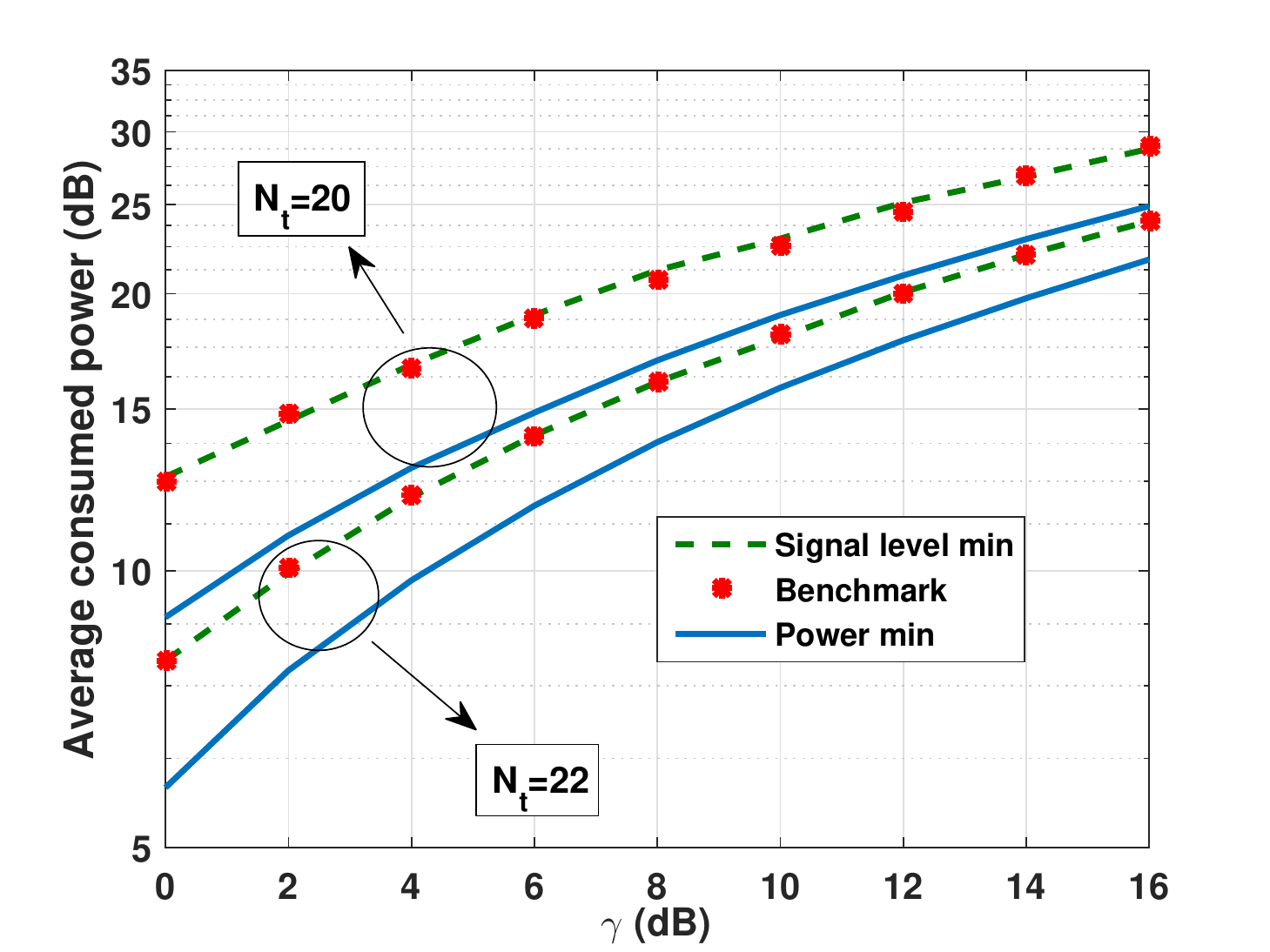}
	\caption{Average consumed power with respect to required SNR for our designed precoders and the benchmark scheme when $N_{U}=19$.}
	\label{fig:pow:snr}
\end{figure}
\begin{figure}[]
	\centering
	\includegraphics[width=9cm]{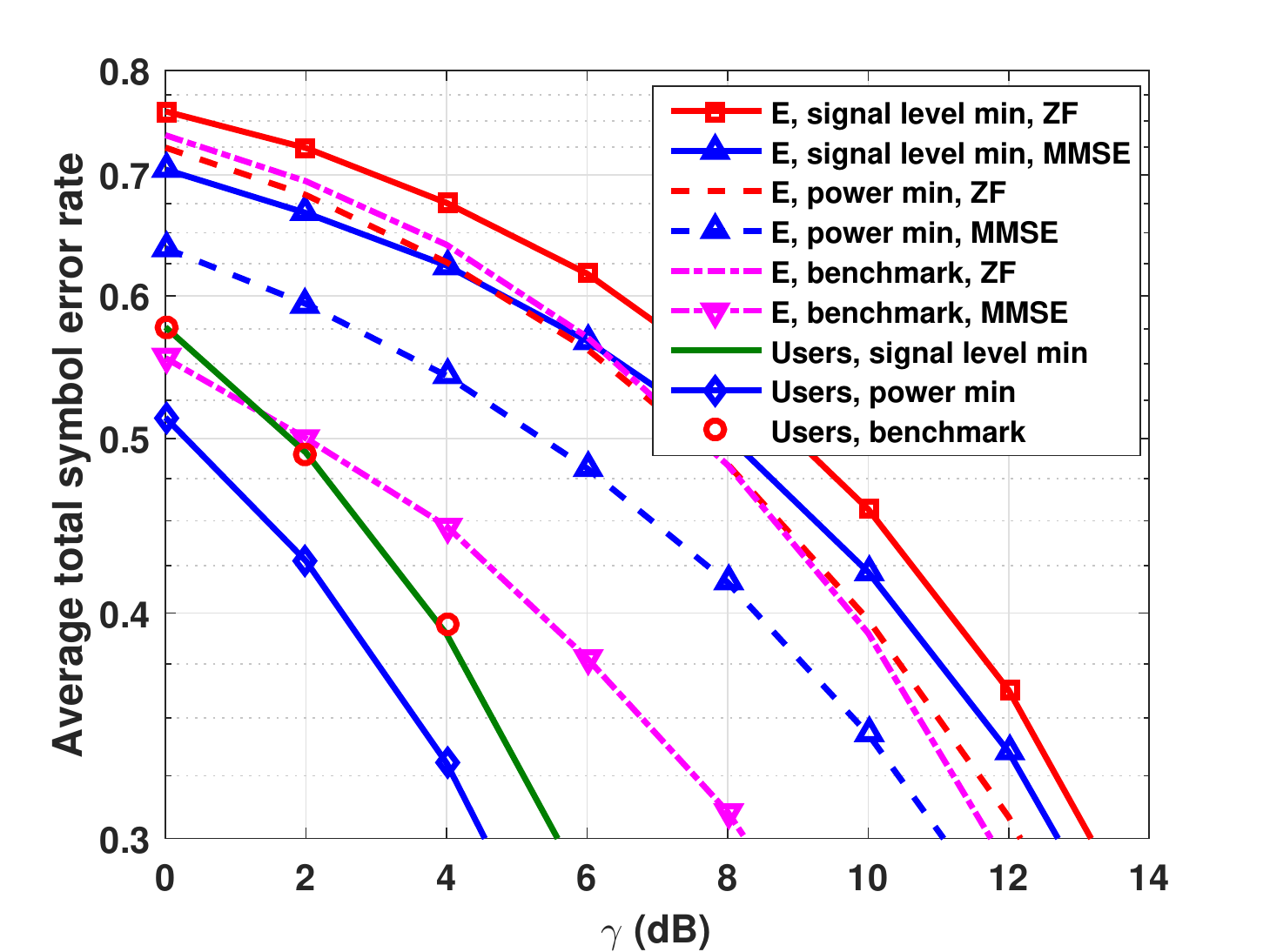}
	\caption{Average SER versus required SNR for our designed precoders and the benchmark scheme when {\rc $N_t=15$,  $N_e=17$ , and $N_{U}=14$}.}
	\label{fig:ser:snr}
\end{figure}
\begin{figure}[]
	\centering
	\includegraphics[width=9cm]{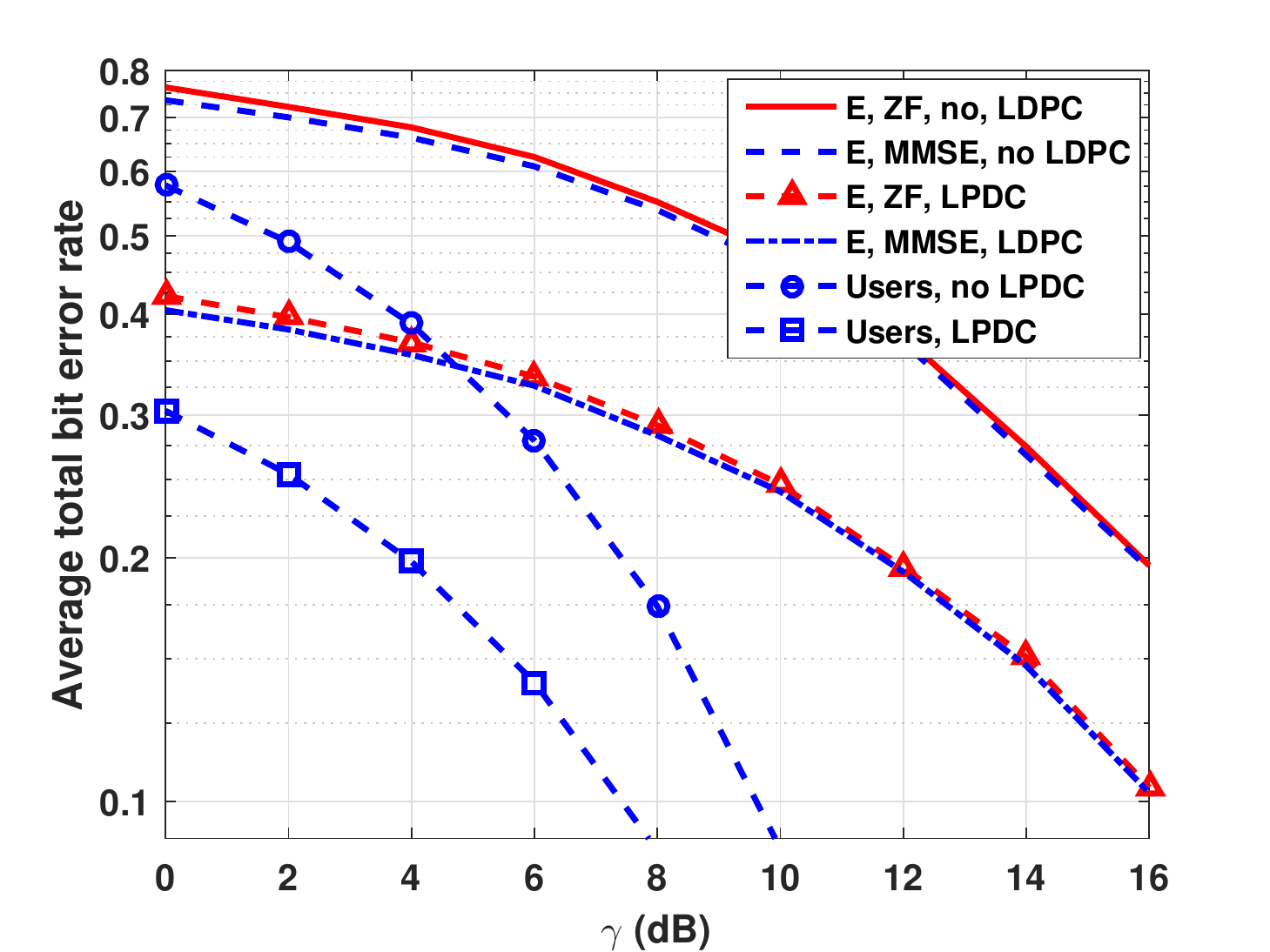}
	\caption{Average BER versus required SNR for {\rc the signal level min precoder when $N_t=10$, $N_e=11$, and $N_U=10$ with the code rate $5/6$.}}
	\label{fig:ser:snr:LDPC}
\end{figure}
\begin{figure}[]
	\centering
	\includegraphics[width=9cm]{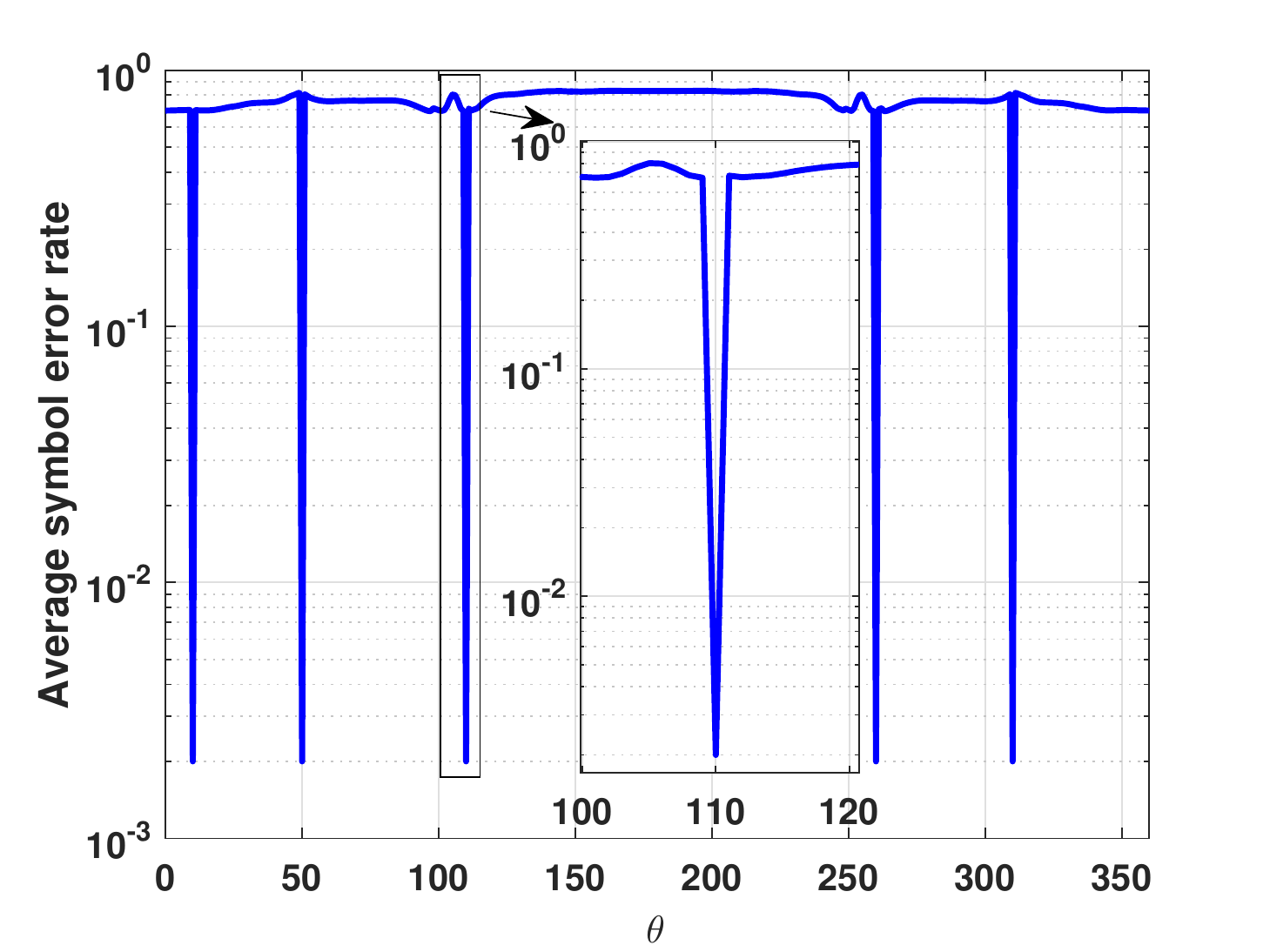}
	\caption{\rc {Average SER versus the location in degrees the signal level min precoder when $N_t=5$ and $T$ communicates with five single-antenna users, i.e., $N_1=N_2=N_3=N_4=N_5=1$.}}
	\label{fig:ULA}
\end{figure}
\begin{figure}[]
	\centering
	\includegraphics[width=9cm]{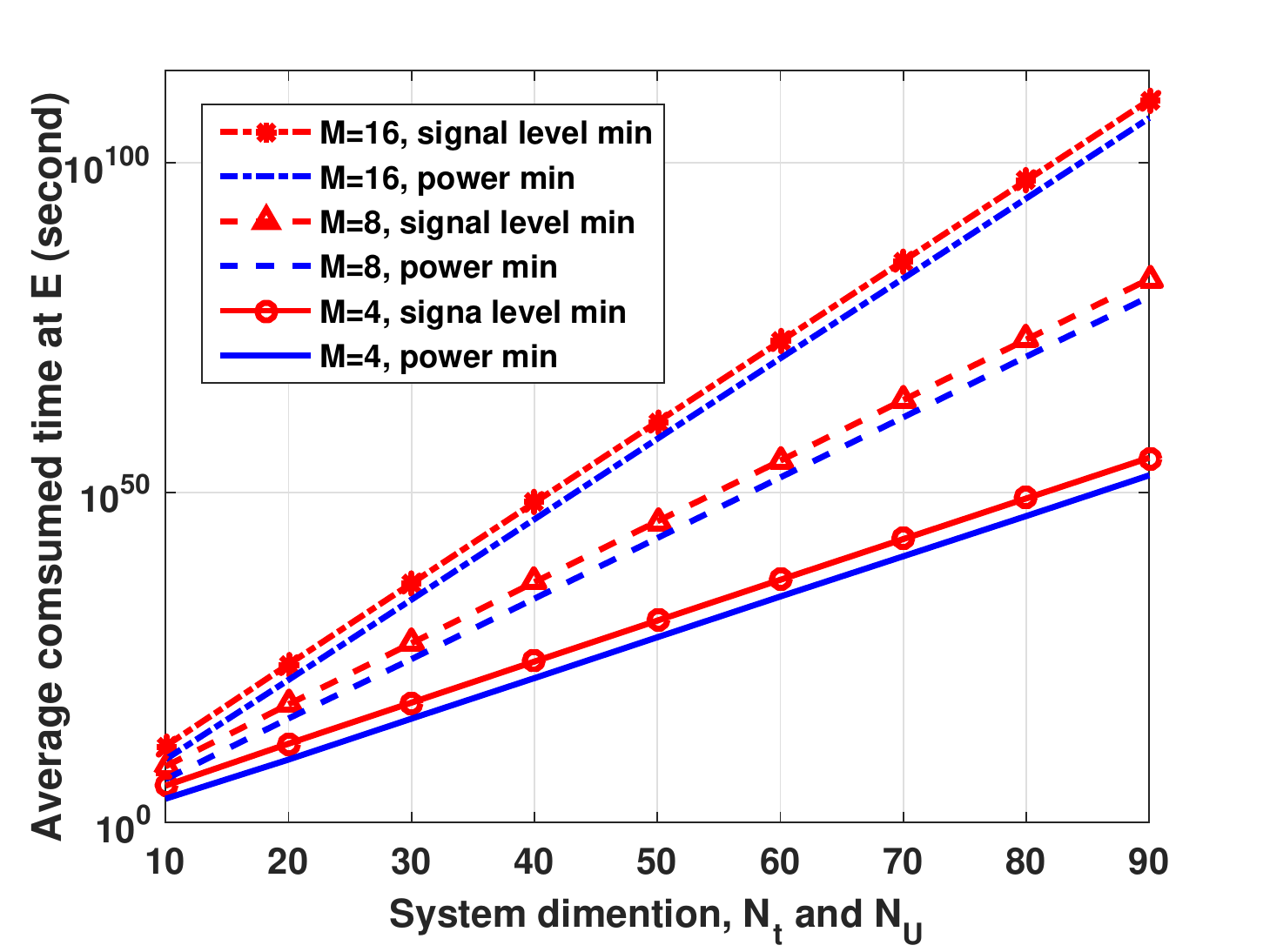}
	\caption{\rc Average consumed time at the eavesdropper with respect to the system dimension when using the proposed precoders with different modulation orders, $N_e=N_t$.}
	\label{fig:eve:time}
\end{figure}
The effect of low-density parity-check (LDPC) codes on the average total bit error rate (BER) at the users and the average BER at $E$ is shown in Fig.~\ref{fig:ser:snr:LDPC} when the signal level minimization precoder is used for the case $N_e \ge N_t$. 

{\rc Next, similar to~\cite{Daly:2009}, we consider a LOS channel and use a uniform linear array (ULA). In this scenario, five single-antenna users are located on the circumference of a circle with radius $4$ m in the angles $10^{\circ}, 50^{\circ}, 110^{\circ}, 260^{\circ}, 310^{\circ}$. The SER with respect to direction of transmission, $\theta$, is shown in Fig.~\ref{fig:ULA}. As we see, the SER sharply decreases to $2 \times 10^{-3}$ from $0.6$ in the locations that users are present.}     

{\rc In the last scenario, we quantify the required time at $E$ to perform brute-force method mentioned in Section~\ref{subsec:bru:for} over all the possible communicated symbols between the transmitter and the receiver. The average brute-force time at $E$ for the proposed precoders using an ordinary computer is shown in Fig.~\ref{fig:eve:time} for different modulation orders. As we see, increasing the system dimension or the modulation order increases the brute-force consumed time enormously.}
\section{Conclusions}
\label{sec:con}
We used the directional modulation technology and followed a signal processing approach to enhance the security over multiuser MIMO channels in the presence of a multi-antenna eavesdropper. {\rc We studied the feasibility of different MIMO receiving algorithms at the eavesdropper and showed that the eavesdropper is able to use the ZF and MMSE algorithms to estimate the users' symbols. The legitimate users can directly decode the received signal via the conventional detectors, e.g., ML, while the results show that the usage of ZF or MMSE causes much more SER at the eavesdropper compared to the users.} In addition, we derived the necessary condition for the feasibility of the optimal precoder for the directional modulation. We proposed an iterative algorithm and non-negative least squares formulation to reduce the design time of the optimal precoders. The results showed that in most of the cases, our designed directional modulation precoders impose a considerable amount of SER on the eavesdropper compared to the conventional precoding. This is due to the fact that our precoders depend on both the CSI knowledge and the symbols while the conventional precoder only depends on the CSI knowledge and the eavesdropper can calculate it. The simulations showed that regardless of the number of antennas, the signal level minimization precoder keeps the SER at the eavesdropper on the maximum value, and it consumes the same power as the power minimization precoder when the difference between the number of transmit and receive antennas is above a specific value. 
In addition, the numerical examples showed that both the power and signal level minimization precoders outperform the benchmark scheme in terms of the power consumption and/or the imposed SER at the eavesdropper.   

\end{document}